\documentclass[a4paper]{article}
\usepackage[utf8]{inputenc}
\usepackage{amsmath, amsfonts,amssymb, hyperref}
\usepackage[all]{xy}
\usepackage{tikz}
\usepackage[standard]{ntheorem}
\usepackage{pifont}
\newcommand{\cmark}{\color{green}{\ding{51}}}%
\newcommand{\xmark}{\color{red}{\ding{55}}}
\renewtheorem{proposition}[theorem]{Proposition}
\renewtheorem{lemma}[theorem]{Lemma}
\renewtheorem{corollary}[theorem]{Corollary}
\theorembodyfont{\upshape}
\renewtheorem{example}[theorem]{Example}
\renewtheorem{remark}[theorem]{Remark}
\renewtheorem{definition}[theorem]{Definition}
\usetikzlibrary{arrows}
\usetikzlibrary{calc}
\usetikzlibrary{hobby}
\usepackage{authblk, url,breakurl,listings,wrapfig,todonotes,nicefrac}
\usepackage{subcaption}
\usepackage[normalem]{ulem}
\newcommand{\usefirst}[1]{\operatorname{uf}(#1)}
\newcommand{\corn}{C-CoRn}

\newcommand{\one}{\mathbf 1}
\newcommand{\demph}{\textbf}
\newcommand{\ask}{\operatorname{?}}
\newcommand{\answer}{\operatorname{!}}

\newcommand{\bool}{\mathbb{B}}
\newcommand{\true}{\mathrm{true}}

\newcommand{\false}{\mathrm{false}}

\newcommand{\fst}{\operatorname{fst}} 
\newcommand{\snd}{\operatorname{snd}}
\newcommand{\inl}{\operatorname{inl}}
\newcommand{\inr}{\operatorname{inr}}
\newcommand{\option}{\operatorname{opt}}
\newcommand{\some}{\operatorname{Some}}
\newcommand{\none}{\operatorname{None}}
\newcommand{\seq}{\operatorname{seq}}

\newcommand{\cat}[1]{+{}\!\!\! {}+{#1}}

\newcommand{\A}{\mathbf A}

\newcommand{\Q}{\mathbf Q}

\newcommand{\QQ}{\mathbb Q}
\newcommand{\B}{\mathcal B}
\newcommand{\BB}{\mathbb B}
\newcommand{\KK}{\mathbb K}

\newcommand{\X}{\mathbf X}
\newcommand{\Y}{\mathbf Y}

\newcommand{\NN}{\mathbb N}
\newcommand{\RR}{\mathbb R}

\newcommand{\mto}{\rightrightarrows}
\newcommand{\dom}{\operatorname{dom}}
\newcommand{\size}[1]{|#1|}
\newcommand{\abs}[1]{|#1|}
\newcommand{\str}{\mathbf}

\newcommand{\Coq}{\textsc{Coq}}
\newcommand{\Incone}{\textsc{Incone}}

\newcommand{\Agda}{\textsc{Agda}}

\newcommand{\sign}{\operatorname{sign}}

\makeatletter
\renewcommand\AB@affilsepx{\hfill \protect\Affilfont}
\makeatother

\title{Continuous and monotone machines}
\author[1]{Michal Konečný}
\author[2]{Florian Steinberg}
\author[3]{Holger Thies}
\affil[1]{Aston University, Birmingham}
\affil[2]{INRIA, Saclay}
\affil[3]{Kyushu University, Fukuoka}
\date{}
\begin{document}
\maketitle
\begin{abstract}
  We investigate a variant of the fuel-based approach to modeling diverging computation in type theories and use it to abstractly capture the essence of oracle Turing machines.
  The resulting objects we call continuous machines.
  We prove that it is possible to translate back and forth between such machines and names in the standard function encoding used in computable analysis.
  Put differently, among the operators on Baire space, exactly the partial continuous ones are implementable by continuous machines and the data that such a machine provides is a description of the operator as a sequentially realizable functional.

  Continuous machines are naturally formulated in type theories and we have formalized our findings in \Coq{}.
  Continuous machines, their equivalence to the standard encoding and correctness of basic operations are now part of \Incone{}, a \Coq{} library for computable analysis.
  While the correctness proofs use a classical meta-theory with countable choice, the translations and algorithms that are proven correct are all fully executable.  
  Along the way we formally prove some known results such as existence of a self-modulating moduli of continuity for partial continuous operators on Baire space.

  To illustrate their versatility we use continuous machines to specify some algorithms that operate on objects that cannot be fully described by finite means, such as real numbers and functions.
  We present particularly simple algorithms for finding the multiplicative inverse of a real number and for composition of partial continuous operators on Baire space.
  Some of the simplicity is achieved by utilizing the fact that continuous machines are compatible with multivalued semantics.
  We also connect continuous machines to the construction of precompletions and completions of represented spaces, topics that have recently caught the attention of the computable analysis community.
\end{abstract}

\section{Introduction}

The main goal of this paper is to add to the tools available for producing correct and efficient software with strict specifications that involve high-level mathematical concepts.
Such methods are required, for example, for reliable simulations of safety-critical physical systems, and the number of applications is steadily growing.
Computable analysis is a formal model for reliable computation involving real numbers and other spaces of interest in analysis.
It extends classical computability theory from discrete structures to continuous ones,
replacing natural numbers as codes for abstract objects by elements of Baire space.
Computable analysis originated with Turing's fundamental work \cite{turing_computable_1936}.
Later, it was extended to a theory of computation on real numbers and real functions by Grzegorczyk \cite{MR0089809} and Lacombe \cite{MR0105357}, and to more general spaces by Kreitz and Weihrauch \cite{kreitz1985theory,weihrauch_computable_2000}.
Algorithms from computable analysis come with a mathematical correctness proof by design, making them well-suited for applications where correctness is essential.
As the semantics of such algorithms tend to be subtle, formal methods can provide an additional reassurance that the produced software is correct.

The present work arose as part of our effort to contribute to the development of a framework for conveniently formulating algorithms from computable analysis in a setting that is both fully computational and features formal correctness proofs.
Our work has meanwhile been used to extend \Incone{} \cite{steinberg2019quantitative}, a library for computable analysis based on the proof assistant \Coq{}.
The task of formalizing known results may not be particularly creative, but difficulties encountered in such an endeavour often lead to new developments.
In the present case, attempts to avoid overly heavy use of \Coq{}'s dependent type system, and to maintain executability within \Coq{} in the presence of non-computational axioms have lead us to concepts that we believe to be of theoretical interest.
It should be kept in mind that the theoretical results are backed by a formal development and we consider the formalization of known and unknown results a part of our contribution.

\paragraph{Main results}
We introduce ``continuous machines'' as an encoding of partial continuous operators
derived from the fuel-based approach to modeling diverging computation in intuitionistic type theories \cite{bove2016partiality}.
Continuous machines can be understood as an abstraction of oracle machines as used to introduce the model of computation central to computable analysis.
There are two main points that support this analogy and distinguish our approach from uses in type theory:
The first is the presence of a functional parameter that is considered an input and that takes the role of the oracle in an oracle machine.
Machines are type-two objects, which is crucial as it makes continuity and information theoretic arguments applicable.
The second particularity is a curried discrete input, meaning that for fixed functional input we get a function that we consider the return value if it is total, if it is not total the return value is undefined.
As a consequence, the natural domains need not be open but only $G_\delta$ sets.

These two features reflect that we really encode continuous operators, i.e.\ partial functions from Baire space to Baire space, as opposed to partial continuous functionals.
We consider operator composition a natural operation, while for functionals the same operation would be called functional substitution and is of lesser importance \cite{constable1973type}.
The emphasis on operators is in tune with the principle of computable analysis to almost consistently replace the natural numbers by Baire space as the base type.
The reason this works is that partial continuous operators can be encoded as elements of Baire space by application of a partial combinatory algebra \cite{kleene1959constructivity,MR0106838}.
A partial operator is computable by an oracle machine if and only if there is a computable code.
Computability as a functional is also equivalent \cite{normann2000computability}.
In abstract terms the whole setting can be captured very concisely as investigating the computable fragment of the realizability topos over Kleene's $K_2$ and should be distinguished from settings where only the full version or that over $K_1$ are the central objects \cite{Bauer:2000:RAC:933370}.

Working directly with codes from Baire space is tedious and our main result, Theorem \ref{resu: main}, shows that continuous machines are completely equivalent:
It provides a full translation from a continuous machine to a code from Baire space and back that preserves computability.
As type-two objects continuous machines are a high-level concept and more convenient for defining partial operators.
We illustrate just how concisely algorithms on continuous data can be formulated using continuous machines at the example of inversion of a real number in Section \ref{sec: inversion}.
As another example we describe a simple and fairly efficient implementation of the composition of partial operators encoded as continuous machines in Section \ref{sec: monotone machines}.
We have fully automated the translations between continuous machines and Baire-space codes in that we have defined them in \Coq{} and provide complete formal proofs of correctness.
In fact, also the two main examples, all other important points made in this paper, and further examples whose description we omit for space reasons, have been formalized\footnote{We have setup a web page with instructions on how access our formal development and relate it to the contents of the paper \url{https://holgerthies.github.io/continuous-machines/}}.
This should be kept in mind as it justifies cutting down on some details for the sake of communicating the important ideas (especially in Section \ref{sec: continuous machines}).

\paragraph{Related work}
The topics of this paper can be viewed from a number of different perspectives.
Clearly there are links to type theory, in particular the results have been formalized in the type theory based proof assistant \Coq{}.
As stated above, the main object of investigation can be understood as a variant of the fuel-based approach to modeling divergent computations in constructive type theories.
A survey of such methods presented in type-theoretic language can be found in \cite{bove2016partiality} together with many relevant references.
However, here we avoid a type-theory-like presentation and prefer mathematical notation that consists of a mix of conventions that are commonly found in some works from computable analysis such as \cite{pauly2016topological, schroeder_phd, Bauer:2000:RAC:933370, kawamuraphd} and the game-centered parts of higher-order computability and programming language theory \cite{longley2015higher}.
In particular we choose to illustrate the use of fuel with Turing machines and oracle Turing machines directly to avoid pointers to type theory in the body of the paper.

The connection of our work to higher-order computability theory is reflected in a possible interpretation of the main result: we provide yet another characterization of the sequentially realizable functionals \cite{longley2002sequentially}.
The connection to computable analysis is also evident and it is our main source of examples.
Some of these examples nicely illustrate connections to precompleteness, constructions forcing precompleteness and completions.
These concepts have quite some history but have recently been rediscovered for their applications in the theory of Weihrauch reductions and in complexity theory for computable analysis \cite{dzhafarov2017joins, brattka2019completion, kawamura2014function}.
Completions can be understood as quotients of coinductive types constructed from the delay monad and have a distinctively domain-theoretic flavour \cite{Capretta05,altenkirch2017partiality}, but we chose not to pursue these aspects at present.

Partial operators on Baire space can also be captured in \Coq{}'s type theory, where partiality is reflected in the use of sigma-types as inputs.
That is, partial operators take as input a dependent pair of the actual input and a proof that this input is from its domain.
While continuous machines provide additional information over a direct definition in \Coq{}, 
we believe that on the computability level this difference is irrelevant as the information can be read off from each respective \Coq{} term.
For an equivalent of a fragment of the \Coq{} terms, an extraction of the additional information a continuous machine provides has been formalized in \Agda{} \cite{xu2013constructive}.
We know that an internal formulation of such a result that covers all definable functions in \Coq{} is impossible.
It involves extracting a modulus of continuity and there are known obstructions to extracting this information extensionally \cite{hotzel2015inconsistency}.
The support of tactics and \Coq{}'s formalization of \Coq{}, i.e.\ the MetaCoq project \cite{metacoq2019hal} should allow to adapt the work in \Agda{} and extend the extraction to cover all or at least the majority of relevant \Coq{} terms.

\paragraph{Notation}

For a set $\Q$ denote the set of finite lists $(q_0, \ldots, q_n)$ with $q_i \in\Q$ by $\seq \Q$.
Let $\epsilon$ be the empty list, \hbox{${}\cat{}$} the concatenation.
The disjoint union of $\Q$ with another set $\A$ denote by $\Q + \A$ and the set of ordered pairs by $\Q \times \A$.
We use $\inl$ and $\inr$ for the inclusions into the sum and for the product $\fst$ and $\snd$ for the projections.
Let $\option \Q$ be the union of $\Q$ with a single new element $\none$ and use $\some q$ for the element of $\option \Q$ corresponding to $q \in \Q$.
Although the above notation is consistent with \Coq{}'s notation, we generally prefer mathematical over type-theoretic notation.
I.e., we use $\in$ instead of $:$, and we use $:$ to separate a quantifier from the body of a formula.
This is, unless we speak about membership in a set of functions, in which case we sometimes write $\varphi \in \A^\Q$ as $\varphi\colon \Q \to \A$.
To indicate that a function is partial, we write $\varphi\colon {\subseteq\Q} \to \A$.
Note that each function $\varphi\colon \Q \to \option(\A)$ can be considered a partial function, but it should be considered a special partial function.
For a given subset $A\subseteq \Q$ we denote by $\varphi|_A$ the restriction of $\varphi$ to $A$.
We slightly abuse notation by sometimes identifying lists and the set of their elements.

\section{Computable analysis revisited}\label{sec: computable analysis}

Let $\Q$ and $\A$ be two sets that we understand to consist of \demph{questions} and \demph{answers}.
We always assume these sets to be countable, and in all concrete examples considered specifying explicit bijections with the natural numbers is straightforward.
In the following we restate standard definitions from computable analysis where we insert $\Q$ and $\A$ for the appropriate copies of $\NN$.
A \demph{representation} of a set $X$ is a partial surjective function $\delta\colon{\subseteq \A^\Q} \to X$.
For $x \in X$, each $\varphi\colon \Q \to \A$ with $\delta(\varphi) = x$ is called a \demph{name} of $x$ and should be understood to provide on-demand information about $x$.
I.e. if $\varphi$ is a name of $x$ then given a question $q\in\Q$ about $x$ the value $\varphi(q) \in \A$ is a valid answer to the question.
Call $\B := \A^\Q$ the \demph{space of names} of the representation, $\B$ due to the case $\Q = \NN = \A$ where it is the \demph{Baire space}.

A \demph{represented space} is a pair $\X:= (X, \delta_\X)$ where $\delta_\X$ is a representation of $X$.
The representation induces topological and computability structures on the set $X$.
Namely, $X$ can be made a topological space by considering the final topology of the representation and an element of a represented space is called \demph{computable} if it has a computable name.
The latter of course presumes that $\Q$ and $\A$ are such that it is clear what computability of a function from $\Q$ to $\A$ means; which is in particular the case when $\Q$ and $\A$ come with explicit bijections to $\NN$.
More generally, $\Q$ and $\A$ can be thought of as being equipped with a numbering.
If a topological space is given and a representation is to be constructed, then the candidates are expected to reproduce the given topology.

\begin{example}[Discrete spaces]
  \label{ex: discrete}
  Whenever $X$ is a discrete countable space such as the Booleans, natural numbers, integers or rationals, the following representation is appropriate:
  Pick $\Q := \one = \{\star\}$ to be the canonical one-element set, $\A := X$ to be the space itself, and let $\varphi\colon \one \to X$ be a name of $x\in X$ if and only if $\varphi(\star) = x$.
\end{example}
A similar idea can be used if $X$ comes with a numbering.
More interesting examples are spaces of continuum cardinality, that cannot be appropriately captured using numberings.

\begin{example}[$\RR_\QQ$: Reals via rational approximations]\label{ex: rational reals}
  One possible way to represent real numbers is via rational approximations.
  To make this formal, first choose the question and answer sets to both be the rational numbers, where a rational question $\varepsilon > 0$ is interpreted as an accuracy requirement and a rational answer as an approximation.
  More precisely, a name of a real number $x\in \RR$ is a function $\varphi$ such that for each rational $\varepsilon > 0$ the value $\varphi(\varepsilon)\in \QQ$ is an $\varepsilon$-approximation to $x$.
  This representation is called the \demph{rational representation $\delta_{\RR_\QQ}$} of the real numbers, and as a partial function $\delta_{\RR_\QQ}\colon{\subseteq \QQ^\QQ} \to \RR$ it is uniquely specified by
  \[
  \delta_{\RR_\QQ}(\varphi) = x  \quad \iff \quad \forall \varepsilon > 0\colon \abs{x - \varphi(\varepsilon)} \leq \varepsilon.
  \]
  We denote the corresponding represented space by $\RR_\QQ$ and use it as one of the running examples.
\end{example}
The represented space $\RR_\QQ$ is widely considered to provide the ``correct'' computability structure on the real numbers and sometimes even used as a benchmark representation in work that reasons about complexity in the setting of computable analysis \cite{MR1137517,lambov2006basic}.
The rational representation is fairly convenient:
It provides a simple question and answer structure and an intuitive interface for accessing information about real numbers.
It only uses a single additional type, namely $\QQ$, which has a well-developed theory in \Coq{}'s standard library.
A closer inspection of the construction of the rational representation reveals that it uses only the structure of the real numbers as a metric space and a distinguished sequence of simple elements.
A generalization applies to a wide variety of spaces of importance in analysis and functional analysis.

Besides spaces of continuum cardinality, the methods of computable analysis can be used to operate on finite spaces with non-discrete topology (see Example~\ref{ex: sign function} below).

\subsection{Continuous and computable functions}\label{sec: continuity}

Fix some represented spaces $\X$ and $\X'$.
Let $\B := \A^\Q$ be the space of names of the representation $\delta_\X$ of $\X$ and $\B':=\A'^{\Q'}$ that of $\delta_{\X'}$.
Let us think of $\Q$, $\A$, $\Q'$ as $\A'$ discrete spaces and consider the induced notion of continuity of operators $F\colon {\subseteq \B} \to \B'$.
More concretely an operator is continuous in this sense if its return values are determined by a finite number of values of its input function.
That is, if for all $\varphi \in \dom(F)$ and each $q' \in \Q'$ there exists a finite list of questions $\str q \in \seq(\Q)$ such that
\[ \forall \psi \in \dom(F)\colon \varphi|_{\str q} = \psi|_{\str q} \implies F(\varphi)(q') = F(\psi)(q'). \]
Since $\Q$ and $\Q'$ are countable, equivalent definitions can be obtained by introducing metric structures on $\B$ and $\B'$ or by requiring a continuous function to preserve limits of sequences.
These equivalences are useful for abstract reasoning about continuity and formal versions are available in the \Incone{} library \cite{steinberg2019quantitative}.
In the case where all question and answer sets coincide with the natural numbers, computability of operators can be introduced by means of oracle machines.
An oracle machine is a Turing machine with a marked oracle query and answer states and a marked oracle tape.
The run of such a machine on oracle $\varphi \in \B$ is defined as the run of a regular machine with the adaption that any time the oracle query state is entered, the content $q$ of the oracle tape is replaced by $\varphi(q)$ and the state is changed to the answer state.
For some background and more details we point the reader to \cite{kawamuraphd}.
It is important to keep in mind the oracle is considered an input to the computation and despite the name and other applications of the same concept where the oracle is fixed, this makes oracle machines a realistic model of computation.
Computability is a refinement of continuity in that any computable operator is continuous.
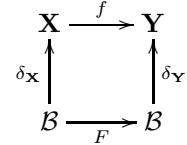
\begin{wrapfigure}{R}{.2\textwidth}
  \centering
  \vspace{-2ex}
  \xymatrix{
    \X \ar[r]^{f} & \Y \\
    \B \ar[r]_F \ar[u]^{\delta_\X} & \B \ar[u]_{\delta_\Y}
  }
    \caption{\\$F$ realizes $f$}
    \label{fig: realizers}
  \vspace{-2ex}
\end{wrapfigure}

The notions of computability and continuity of partial operators on Baire space can be pushed forward through representations to apply to functions between represented spaces:
An operator $F\colon{\subseteq \B} \to \B'$ is said to \demph{realize} a function $f\colon\X \to \X'$ between represented spaces if for each name $\varphi$ of some $x \in \X$ the value $F(\varphi)$ is defined and a name of $f(x) \in \X'$ (c.f. Figure \ref{fig: realizers}).
The function $f$ is called \demph{continuous} or \demph{computable} if it has a realizer with that property.
In all cases we are interested in, this notion of continuity coincides with topological continuity with respect to the natural topology on the space which is in turn reproduced by the final topology of the representation.
Without going into details, let us remark that this is because all representations that we consider are \demph{admissible} \cite{schroeder_phd}.
For instance, the rational representation introduced in Example~\ref{ex: rational reals} is admissible, and consequently a function on the real numbers is $\varepsilon$-$\delta$-continuous if and only if it is continuous with respect to this representation.  
These and similar statements have been formally proven in \Incone{} \cite{steinberg2019quantitative}.

Continuity is a prerequisite for computability.
The real numbers are connected, discrete spaces are totally disconnected and images of connected sets under continuous functions are connected.
For this reason most functions from the reals to the Booleans fail to be computable.
This is for instance true for equality checks, comparisons and other operations that are routinely used and seem indispensable for applications in numerics.
Often, computability can be recovered by replacing a discrete target space by an appropriate non-discrete finite space.
\begin{example}[Sign function and Kleeneans]\label{ex: sign function}
  The sign function is discontinuous as a function from the reals to a discrete space (for instance to its image as a subspace of the real numbers).
  A computable version can be recovered by replacing its three possible values with elements of the following space:
  Consider the three-point set $\{\true_\KK, \false_\KK, \bot_\KK\}$ and equip it with a representation $\delta_\KK$ defined on names $\varphi$ of type $\NN \to \option \BB$ by
  \[
  \delta_{\KK}(\varphi) = \begin{cases}
    b_{\KK}  & \text{if there exists some $n$ such that } \varphi(n) = \some b \\
    & \text{and for all } m < n, \varphi(m) = \none \\
    \bot_{\KK} & \text{ otherwise.}
  \end{cases}
  \]
  That is, the constant $\none$ sequence is a name of $\bot_\KK$ and for any other sequence the first element that is not $\none$ determines which of $\true_\KK$ and $\false_\KK$ is named.
  
  We refer to $\KK := (\{\true_\KK, \false_\KK, \bot_\KK\}, \delta_\KK)$ as the \demph{Kleeneans}.
  This space models the behavior of three-valued logics considered by Kleene, hence the name.
  Note that the representation is total, i.e.~all sequences are valid names, which makes it convenient to define realizers of functions into the space.
  When defining functions that use the Kleeneans as an argument, it is often more convenient to require names to be monotone in the sense that if the sequence contains $\some b$, all subsequent its elements repeat this value.
  The use of this restriction does not change the space, as an arbitrary name can be computably transformed into a monotone name.
  
  The sign function as a function from the reals to the Kleeneans can be defined from the Boolean comparison on the reals as
  \[ \sign_\KK(x) := \begin{cases} (0 < x)_\KK & \text{if } x \neq 0 \\ \bot_\KK &\text{otherwise.} \end{cases} \]
  Where the strict inequality could as well have been replaced by non-strict inequality as the case $x = 0$ is treated separately anyway.
  A continuous realizer of the sign as a function of type $\RR_\QQ \to \KK$ can be specified from the Boolean comparisons on the rational numbers as
  \[ F(\varphi)(n) :=
  \begin{cases}
    \some (0 < \varphi(2^{-n})) & \text{if } |\varphi(2^{-n})| > 3 \cdot 2^{-n} \\
    \none & \text{otherwise.}
  \end{cases} \]
  As comparison of rational numbers is decidable, this realizer is computable.
  To verify its correctness note that if $\varphi$ is a name of $0$, then $|\varphi(2^{-n}) - 0| \leq 2^{-n}$ implies that $F$ returns the constant $\none$ sequence.
  If $\varphi$ is a name of some $x \neq 0$ then there exists some $n$ such that $2^{-(n-2)} < |x|$ and thus $|\varphi(2^{-n})| > 3 \cdot 2^{-n}$ by a use of the reverse triangle inequality.
  Whenever we are in the first case it follows that as Booleans $0 < \varphi(2^{-n}) = 0 < x$.
  In combination of these we conclude that $F$ returns a name of the correct value.

  The requirement to be greater than $3 \cdot 2^{-n}$ in the definition of $F$ can be replaced by just demanding the same value to be greater or equal $2^{-n}$ while maintaining correctness.
  However, the former forces that the realizer always returns names that are monotone in the sense discussed above.
  To verify this note that as $|\varphi(2^{-n}) - \varphi(2^{-(n+1)})| \leq 3 \cdot 2^{-(n+1)}$, whenever $|\varphi(2^{-n})| > 3 \cdot 2^{-n}$ it follows that $|\varphi(2^{-(n+1)})| > 3 \cdot 2^{-(n+1)}$.
\end{example}

\subsection{Isomorphy, equivalence and precompleteness}\label{sec: precompleteness}
Represented spaces $\X$ and $\Y$ are \demph{isomorphic} if there exists a continuous bijection with continuous inverse and \demph{computably isomorphic} if there exists a computable bijection with computable inverse.
Two representations of the same space are called \demph{equivalent} if the identity function is an isomorphism between the corresponding represented spaces and in this case we call realizers of the identity function \demph{translations} between the representations.
The total representation of the Kleeneans is equivalent to its restriction to monotone names and the corresponding represented spaces are therefore isomorphic.

\begin{example}[Booleans as Kleeneans]\label{ex: booleans}
  Intuitively, the Kleeneans are an extension of the Booleans with an additional element for divergence.
  Formally, this is reflected in the fact that the Booleans are isomorphic to a subspace of the Kleeneans.
  Any subset of a represented space can be turned into a represented space by equipping it with the co-restriction of the representation.
  Let $\BB_\KK$ denote the represented space obtained by co-restricting the representation $\delta_{\KK}$ to the set $\{\true_\KK, \false_\KK\}$.
  Recall that, according to Example \ref{ex: discrete} a name of $b\in \BB$ is a function $\varphi \colon \{\star\} \to \BB$ on the canonical one point set that returns the element itself, i.e.\ such that $\varphi(\star) = b$.
  The inclusion $b \mapsto b_\KK$ of the Booleans into the Kleeneans can thus be realized by the computable operator defined by $F(\varphi)(n) := \some (\varphi(\star))$.
  The inverse can be realized by searching through a given name $\varphi\colon \NN \to \option(\bool)$ of an element of $\KK$ for the first $n$ such that $\varphi(n) \neq \none$ and returning its value, i.e.~returning $b$ if $\varphi(n) = \some b$.
  This defines a properly partial function as the algorithm diverges for the constant $\none$ function, i.e.\ if the input is the name of $\bot_\KK$.
  On all other inputs, in particular on the inputs that are names of an element of $\BB_\KK$, it converges and returns correct values.
\end{example}
While the spaces $\BB_\KK$ and $\BB$ are isomorphic, there are significant differences between their representations.
The representation of $\BB$ has the property that for any element of the space $b\in \BB$ and any question $q$, the set of answers can be separated into correct and incorrect ones in such a way that any function that returns correct answers is a valid name.
(The Booleans may be too simple an example; for a richer illustration of this property, see the rational representation from Example~\ref{ex: rational reals}.)
Since this property implies that the domain of the representation is closed, the representation of $\BB_\KK$ does not satisfy it.
Indeed, whether an answer in the representation of $\BB_\KK$ should be considered correct can not be made sense of without knowing at least some of the answers to other questions.
The representation of $\BB_\KK$ has the somewhat orthogonal property that it is \demph{precomplete}:
any computable function to $\BB_\KK$ has a total computable realizer.
Here, and more generally for any precomplete representation, computability can be replaced with continuity.
Any partial computable function into $\KK$ has a higher-order primitive recursive realizer, which is particularly handy in a formal setting as such a realizer can be defined in any language that is expressive enough to cover basic arithmetic facts.
Thus it allows to avoid a full formalization of oracle machines or another explicit computational model for capturing computability.

As it is instructive for understanding the central topic of this paper, let us include a sketch of how to prove precompleteness of the representation of $\BB_\KK$.
Let $M$ be an oracle machine that computes some partial operator $F$ that realizes a function $f \colon \X \to \BB_\KK$.
This means that for each fixed name $\varphi$ of some $x$, the evaluation of $M$ with $\varphi$ fixed in the oracle slot computes a name $F(\varphi)\colon \NN \to \option(\BB)$ of $f(x)\in\BB_\KK$.
To construct a machine that computes a total realizer of the same function, for a given oracle $\varphi$ on input of $n$ proceed as follows:
Execute $M$ with $\varphi$ in the oracle slot, on input $0$ for $n$ basic computational steps.  
Whenever this computation terminates, start over with the input increased by one but without resetting the timer.
Once the time allowance runs out, return the value that $M$ produced on the biggest input value for which it terminated and $\none$ if no such value exists.
This function is clearly total and computes a realizer: If $F(\varphi)$ is defined and a name of an element of $\BB_\KK$, the modified machine returns a different name of the same element.
If $F(\varphi)$ is undefined the modified realizer defaults to an infinite sequence of $\none$ from the first point where a divergence is encountered.
In particular it returns a name of $\bot_\KK$ whenever the original algorithm fully diverges in the sense that it does not return anything on any input.
If the original function $f$ was partial, the extended realizer can thus be understood to compute a total extension that takes values in $\KK$ instead of $\BB_\KK$.
This makes some authors consider the space $\KK$ as the completion of $\BB_\KK$ and also of $\BB$, for details we point the reader to \cite{brattka2019completion}.

Baire space has a precomplete representation that can be precomposed to obtain a computably equivalent precomplete version of an arbitrary representation.
A direct description of this construction that is particularly instructive for our purposes is as follows:
Given a representation $\delta$ with space of names $\B$, define a new representation $\delta'$  by replacing the questions $\Q$ by $\NN \times \Q$ and the answers $\A$ by $\option(\A)$ and let $\varphi\colon \NN\times \Q \to \option(\A)$ be a name of $x$ if and only if searching for the minimal $n$ such that $\varphi(n) = \some a$ produces a name of $x$ with respect to the original representation $\delta$.
The representation $\delta'$ is computably equivalent to $\delta$ as searching is a computable operation.
The realizer modification from the last paragraph can be adapted to work whenever an arbitrary representation $\delta$ on the target space is replaced by the computably equivalent representation $\delta'$ and, at least if question and answer sets are simple enough, the modified realizer is not only a total computable function but higher-order primitive recursive.

More information about precompleteness can be found in early work about computable analysis \cite{kreitz1985theory, schroeder_phd}.
More recently, the concept has also raised some interest in the theory of Weihrauch reductions \cite{brattka2018weihrauch}.
Precomplete representations are not only useful in formal developments but also tend to be easier to optimize for efficiency.
The rational representation from Example~\ref{ex: rational reals} is not precomplete and a popular precomplete and computably equivalent representation represents real number by sequences of intervals with rational endpoints.
Software packages for exact real arithmetic that aim at efficiency commonly prefer practical adaptations of such a precomplete representations for internal representation of the real numbers.
The rational representation only implicitly appears in handling of in and output of real number data.
This is while theoretical work on complexity theory in the setting of computable analysis tends to prefer representations with a closed domain \cite{schroedert2}.
Indeed, precomplete representations are somewhat incompatible with the established complexity model \cite{kawamura_complexity_2012} and a popular source of counterexamples \cite{kawamura2014function,steinberg2017computational}.
This mismatch has recently lead to some suggestions for adapting Kawamura and Cook's framework \cite{NEUMANN2019}.

\subsection{Formal development in \Incone{}}\label{sec: incone}
We built on the Coq library \Incone{} to formalize our results in \Coq{}.
\Incone{} translates concepts from computable analysis to a formal setting and provides many basic results about these notions \cite{steinberg2019quantitative}.
Continuity of functions on Baire space and continuity of functions between represented spaces as well as concepts relating to multivaluedness in the next section are defined closely following our mathematical description above.
A minor difference between the internals of \Incone{} and the description here is that the question set of a representation is additionally required to come with an explicit inhabitant used as default question.
Until recently the same assumption was made about the answer sets.
We get back to the use of default questions and answers in the last section.

\Incone{} uses the name ``continuity space'' for what resembles a represented space.
The reason for this divergence in terminology is that, in contrast to continuity, \Incone{} does not explicitly define computability.
Instead it approximately captures it via \Coq{}'s Prop/Type distinction:
A proof that a continuous realizer exists is a certificate for continuity while the explicit definition of a realizer as a function on the domain of the representation is a candidate for certifying computability.
However, the assumptions about the question and answer types are chosen too weak to imply definability of a bijection with the natural numbers or decidability of equality on countable types.
Thus, if the question and answer types are abstract in nature and meant to be given computational meaning through an encoding themselves, the concepts of definability of a realizer and computability may diverge, which is the reason for the decision to not use the name ``represented space''.
This could be mended by use of types from the math-comp library \cite{mahboubi2017mathematical}, specifically the countTypes that come with an explicit bijection.
However, such a switch would currently increase the strain on users considerably as constructing corresponding bijections can be arduous, especially for types  suitable for efficient computation.
In applications computational types usually appear as answers and not as questions, but an asymmetry in the assumptions is difficult to maintain when building spaces of functions (see section \ref{sec: function spaces}).

Another reason for not making computability a definition in \Incone{} in the current state is how real numbers are treated.
In the standard library of \Coq{} in its current version, the real numbers are introduced axiomatically and non-computational axioms that break with the Prop/Type distinction are stated as global facts instead of being treated as parameters.
While use of these axioms for correctness proofs should be eligible, their global status makes it possible to misuse them in definitions.
For a user being presented with a dependent type of a function together with a term proving a specification, there is no easy way to tell if all uses of the axioms are appropriate and purely for specification purposes.
That is until he attempts to run an algorithm and the corresponding term fails to fully reduce.
We hope this problem to be solved in the not too distant future, as the classical development of real numbers is currently being overhauled under the name mathcomp-analysis \cite{affeldtmathematical}.
However, many computation heavy libraries like the interval library by Melquiond~\textit{et.~al.} still specify against the real numbers from the standard library.
We would like these to be available for use in our development and have thus not made the switch to mathcomp analysis yet.
The most ``proper'' way to introduce a notion of computability would be to rely on a formalization of a model of computation \cite{ForsterSmolka:2018:Computability-JAR}.
\newpage

\section{Multifunctions and abstract machines}\label{sec: multifunctions}

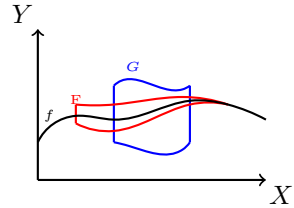
\begin{wrapfigure}{R}{.40\textwidth}
  \begin{tikzpicture}[thick,scale=1]
    \draw[->] (0,0) -- (0,2);
    \draw[->] (0,0) -- (3,0);
    \node at (3.2,-.2) {$X$};
    \node at (-.2,2.2) {$Y$}; 
    \node at (1.25,1.5) {\color{blue} \tiny $G$}; 
    \node at (0.15,.85) { \tiny $f$}; 
    \draw[blue] (1,1.25) to [curve through ={(1,1.25)  . . (1.9,1.2)  }]  (2,1.25);
    \draw[blue] (2,1.25) -- (2,.5);
    \draw[blue] (1,.5) to [curve through ={(1.2,.45)  . . (1.9,.4)  }] (2,.5);
    \draw[blue] (1,1.25) -- (1,.5);
    \draw[red] (.5, 1) -- node[above]{\tiny F} (.5, .75);
    \draw[red] (.5,1) to [curve through ={(1,1)  . . (1.9,1.1)  }] (2.5,1);
    \draw[red] (.5,.75) to [curve through ={(1,.65)  . . (2,1)  }] (2.5,1);
    \draw (0,.5) to [curve through ={(.5,.85)  . . (1,.8) . . (2,1.05) . . (2.5,1) }] (3,.8) ;
  \end{tikzpicture}
  \caption{$f$ chooses through \textcolor{red}{$F$} and \textcolor{red}{$F$} tightens \textcolor{blue}{$G$}, i.e. $\textcolor{red}{F} \prec \textcolor{blue}{G}$.\newline
    Thus $f$ also chooses through \textcolor{blue}{$G$}.}
  \label{fig: tightenings}
\end{wrapfigure}

In computable analysis it is often the case that continuity fails for extensionality reasons and dropping extensionality by using multivalued functions is a popular and powerful tool to work around such problems.
A multivalued function from a set $X$ to another set $Y$ assigns to each element $x\in X$ a set of eligible return values $F(x) \subseteq Y$.
This set may be empty and those $x$ for which it is non-empty are considered to constitute the \demph{domain} $\dom(F) \subseteq X$.
The multifunction is called \demph{total} if its domain is all of $X$ and \demph{single-valued} if it only returns sub-singletons, i.e. each value set has at most one element.
Each partial function can be considered a single-valued multifunction; this multifunction uniquely specifies the partial function and is total if and only if the partial function is.

A partial function $f$ is said to \demph{choose through} a multifunction $F$ if for each $x \in \dom(F)$ it returns an eligible return value, i.e. $f(x)$ is defined and an element of $F(x)$.
Note that this allows the domain of the partial function to be bigger than that of the multifunction.
A multifunction should be considered a specification of all the partial functions that choose through it and this defines an important ordering on the multifunctions:
A multifunction $F$ is said to \demph{tighten} another multifunction $G$, in symbols $F \prec G$, if any partial function that is a choice for $F$ is also a choice for $G$ (c.f. Fig. \ref{fig: tightenings}).
This can equivalently be formulated as $F \prec G$ if and only if
\[ \dom(G) \subseteq \dom(F) \quad\text{and}\quad \forall x\in \dom(G), F(x) \subseteq G(x). \]
For partial functions $f \prec g$ if and only if $f$ is an extension of $g$.
A partial function $f$ chooses through a multivalued $F$ if and only if $f \prec F$.

The multivalued functions from $X$ to $Y$ are in one to one correspondence with the relations, i.e. the subsets of $X \times Y$.
However, multivalued functions should be understood as directed and thus the natural operations differ.
For instance, for $F \colon Y \mto Z$ and $G\colon X \mto Y$ the composition as multivalued functions $F \circ G$ returns on input of $x\in X$ the set
\[ (F \circ G)(x) := \{z \in Z \mid G(x) \subseteq \dom(F) \wedge \exists y\in G(x)\colon z \in F(y)\}. \]
This defines an associative operation that is asymmetric.
By contrast, the natural composition of relations, which has only the existential part of the condition, is symmetric.
Note that the non-existential part of the condition does not mention $z$ and can be understood as a domain-condition reflected in the identity $F \circ G = F \circ_R G|_{\{x \mid G(x) \subseteq \dom(F)\}}$, where $\circ_R$ denotes the relational composition.
The multifunction composition respects the interpretation as specifications in the sense that if partial functions $f$ and $g$ choose through $F$ and $G$ respectively, then their composition as partial functions chooses through the composition of $F$ and $G$ as multifunctions while the relational composition does not.
More generally, the multifunction composition respects the tightening ordering.

\subsection{Multivalued functions and realizability}
Recall that a function $f\colon \X \to \X'$ between represented spaces is realized by some $F\colon {\subseteq \B} \to \B'$ if $F$ translates each name of any input $x\in\X$ to a name of the corresponding return value $f(x)$.
An alternate way to express this is that $\delta_{\X'} \circ F$ is an extension of $f \circ \delta_\X$ as expressed diagrammatically in Figure~\ref{fig: realizers}.
This suggests a lift of being a realizer to multivalued functions: say that a multifunction $g\colon \X \mto\X'$ is realized by another multifunction $G\colon \B \mto\B'$ if $\delta_{\X'} \circ G \prec g \circ \delta_\X$.
Here we just replaced function composition by multifunction composition and being an extension by being a tightening.
This definition behaves as expected and for partial functions reproduces what we have already done in examples:
An operator is a realizer of a partial function if and only if it is a realizer with respect to the subspace representation on the argument space.
For multifunctions, being a realizer is preserved under tightening the realizer and ``loosening'' the realized function.
As continuity and computability are preserved under composition and multifunction composition is compatible with tightenings, the notions of continuous and computable realizability still compose well.

While any multivalued function is uniquely determined by its partial choice functions, the same need not be true when the attention is restricted to continuous partial functions.
A continuously realizable multifunction need not have any partial continuous choice functions at all.
From the point of view of constructive logics, the existence of multifunctions with continuous realizer but without continuous choice functions can be interpreted as failure of a choice principle.
A constructive proof of a forall-exists-statement can more often than not be used to devise a computable realizer of the corresponding total multifunction. 
When this multifunction does not have any continuous choice functions, it means that it is not possible to constructively prove the existence of a function that selects an existential witness.

As we are mostly interested in continuous and computable realizers, one may argue that allowing multivalued realizers is not necessary.
Continuity makes sense only for functions, or at least is known to be problematic in the presence of multivaluedness \cite{Pauly2013RelativeCA}.
However, we shall use a notion of algorithms that can \textit{a priori\/} give multivalued results.
Although it is possible to force single-valuedness, it can be convenient to not always do this right away and the notion of multivalued realizers turns out to be useful.
Another consequence that is useful in other parts of computable analysis is that when multivalued realizers are allowed, any multifunction $g$ between represented spaces has a unique realizer that is maximal with respect to tightenings, namely $\delta_{\X'}^{-1}\circ g \circ \delta_\X$ \cite{brattka2017weihrauch}.

\subsection{Algorithmic content and machines}
\begin{figure}
  \centering
  \begin{subfigure}[t]{.475\textwidth}
    \centering
    \begin{tikzpicture}
      \tikzstyle{every node}=[font=\footnotesize]  
      \node at (-1,3.75) {$\varphi$};
      \draw[->] (-1,3.5) -- (-1,3);
      \draw[dashed,->] (-1,3) -- (-1,2.5) -- (-1.25,2.5);

      \draw (-3.5,3) rectangle (1.5,0);
      \node at (-3.25,2.75) {$M$};
      
      \draw (-2.75,2) rectangle (-.75,1);
      \node at (-2.55,1.8) {$O$};
      \draw[->] (-1.25,1) -- (-1.25,0.75) node[below]{$a'$};
      \draw[->] (-2.25,0.75) -- (-2.25,1);
      \node at (-2.25,0.6) {$q'$};
      
      \draw[rounded corners = 5pt] (-2.25,2.25) rectangle (-1.25,2.75);
      \node at (-1.75,2.5) {$\varphi$};
      \draw[->] (-1.5,2.25) -- (-1.5,2);
      \draw[->] (-2,2.0) -- (-2,2.25);

      \draw[->] (-2,-0.5) node [below]{$(n,q')$} -- (-2,0);
      \draw[dashed,->] (-2,0) -- (-2,0.25) -- (-3,.25) -- (-3,1.5) -- (-2.75,1.5);
      \draw[dashed,->] (-2.25,0.25) -- (-2.25,.35);
      \node (C) at (-1.75,1.5) {run $n$ steps};
      \draw[dashed,->] (-.75,1.5) -- (-.3,1.5);

      \node (K) at (.5,1.5) {terminated?};
      \draw[dashed,->] ($(K)+(0.7,-0.2)$) -- ++(0,-0.5) node[midway,right]{\footnotesize \xmark}
      -- (1.2,.25);
      \draw[dashed,->] (1.2,.25) -- (0,.25) -- (0,0);
      \node at (.5,.4) {\tiny$\none$};
      \draw[dashed,->] ($(K)+(0.7,-0.2)$) -- node[midway,left]{\footnotesize \cmark} ++(-.5,-0.5);
      \draw[dashed,->] ($(K)+(.2,-.7)$) -- ++(-1.5,0) -- ++(-.25,-.2);
      \draw[dashed,->] (-1.25,.35) -- (-1.25,.25) --(0,.25);
      \node at (-.5,.4) {\tiny$\some(a')$};
      \draw[->] (0,0) -- (0,-0.5) node[below]{$o$};
    \end{tikzpicture}
    \caption{
      How to construct an $M$ such that $F_M\prec F$ from an oracle machine $O$ that computes $F\colon {\subseteq\B}\to \B'$.
      Access to $O$ allows to conclude what values of $\varphi$ where accessed and even to express $M$ in G\"odels system-T.
      However, we consider $M$ to not provide this information but just as a function.
      This is emphasized by the single arrow from $\varphi$ to $M$.
      What values of $\varphi$ are accessed by $M$ is considered extra information and discussed in Section \ref{sec: continuous machines}.}
    \label{fig: FM}
  \end{subfigure}
  \hfill
  \begin{subfigure}[t]{.475\textwidth}
    \centering
    \begin{tikzpicture}
      \tikzstyle{every node}=[font=\footnotesize]  
      \node at (-1,3.75) {$\varphi$};
      \draw[->] (-1,3.5) -- (-1,3);
      \draw[dashed,->] (-1,3) -- (-1,2.5);

      \draw[rounded corners = 5pt] (-2.5,3) rectangle (1.5,0);
      \node at (.75,1.75) {$F$};
      
      \draw[rounded corners = 5pt] (-2,2.5) rectangle (0,1.5);
      \node at (-1,2) {$M$};
      \draw[->] (-.5,1.5) -- (-.5,1.25);
      \node at (0.05,1.14) {$o \stackrel{?}\neq\none$};
      \draw[->] (-1.5,1.25) -- node[below]{$(i,q')$} (-1.5,1.5);
     
      \draw[->] (-1.5,-0.5) node [below]{$q'$} -- (-1.5,0);
      \draw[dashed,->] (-1.5,0) -- (-1.5,0.9);
      \node at (-1.85,.25) {\tiny $i=0$};
      
      \draw[dashed,->] (-.225,.85) -- ++(.225,-.5) node[midway,right]{\footnotesize \cmark}
      -- (0,0);
      \node at (.75,.25) {\tiny$o = \some(a')$};
      \draw[dashed,->] (-.225,.85) -- node[midway,left]{\footnotesize \xmark} ++(0,-0.35);
      \draw[dashed,->] (-.25,.5) -- node[below]{\tiny increase $i$} (-1.5,.5);
      \draw[->] (0,0) -- (0,-0.5) node[below]{$a'$};
    \end{tikzpicture}
    \caption{
      Given a function $M$ a simple while loop recovers a choice function of $F_M$.
      The rounded corners of the boxes indicate that $M$ need not be computable.
      If $M$ is computable, then so is $F$ and both boxes can be depicted with pointy corners.
      While $M$ is a total function, the while loop may diverge whenever $\varphi$ is not from the domain of $F_M$ and thus the constructed $F$ can be properly partial.
    }\label{fig: MF}
  \end{subfigure}
  \caption{}
\end{figure}
Now that we discussed the tools we need for specification, the next step is to see how to produce computational objects that can fulfill these specifications.
In particular, we are interested in devising operators, that is, partial functions on Baire space or Baire-space-like spaces of functions.
We take the fuel-based approach for capturing divergence in type theories and adapt it to the operator and oracle machine setting of computable analysis.
As before, fix some countable sets $\Q$, $\A$, $\Q'$ and $\A'$ and abbreviate $\B := \A^\Q$ and $\B' := \A'^{\Q'}$.
To each function $M\colon \B \to \option(\A')^{\NN \times \Q'}$ assign a multifunction $F_M\colon \B \mto \B'$ by
\[ F_M(\varphi) := \{ \psi\in \B' \mid \forall q', \exists n, M(\varphi)(n,q') = \some (\psi(q'))\}. \]
This value-set can be empty or contain more than one element but for each $\varphi$ the set $F_M(\varphi)$ is a closed subset of $\B'$.
Looking back to the discussion in Section~\ref{sec: precompleteness}, up to multivaluedness being involved, a machine can be understood as replacing the representation on the target space by its precompletion.

Whenever an operator $F$ can be computed by an oracle machine, $M$ can be chosen to be the function that on inputs $\varphi$, $n$ and $q'$ runs the oracle machine for up to $n$ time steps on input $q'$ and oracle $\varphi$, in case of termination returns $\some a'$ where $a'$ is what the machine returned and otherwise returns $\none$ (see Fig. \ref{fig: FM}).
Then $F_M$ is single-valued and the corresponding partial function extends $F$ by the very definition of what it means for an oracle machine to compute $F$.
The values of $F$ can be recovered from those of $M$ by searching through increasing values of $n$, and for a general $M$ this gives a choice function of $F_M$ (Fig. \ref{fig: MF}).
If $M$ is obtained from an oracle machine and implemented in a reasonable way, this leads to a quadratic overhead in running times over direct execution which can be further reduced to a constant factor by considering only powers of two for $n$.
\Coq{}'s standard library proves a restricted choice principle called constructive epsilon that can be used to recover a choice function of $F_M$ as a dependently typed function when given an arbitrary function $M$.
Internally this leads to a linear search through the values.

Motivated by the oracle machine example, we call a function $M$ a \demph{machine} for $F$ if $F_M$ tightens $F$.
This analogy should be taken with a grain of salt and does not appropriately reflect the role played by the natural number input $n$.
We refer to $n$ as the \demph{effort parameter}, and while higher values do usually indicate a higher time consumption of the computation, it need not be directly related to the running time.
In particular, we refrain from interpreting the effort parameter as ordering a computation into a sequence of steps (as is the case for oracle machines) and instead embrace the view that it is a functional input.
Finally note that one may also consider the effort parameter as an index and a machine as a sequence of approximate realizers.

Figure \ref{fig: machine} illustrates the relation between a machine $M$ and the operator $F$ it implements. 
As an example, let us discuss the task of finding a multiplicative inverse in the rational representation in some detail.

\begin{figure}
  \centerline{
  \xymatrix@!0@C=2em@R=6em{
    f\colon \subseteq & \X \ar[rrrrrr] &&&&&& \Y \\
    F\colon & \B \ar@<-.5ex>[rrrrrr] \ar@<.5ex>[rrrrrr] \ar[u]^{\delta_\X} &&&&&& \B' \ar[u]_{\delta_\Y} \\
    M\colon & \B \ar[rrrrrr] \ar@{..>}[u]^{id} &&&\ar[u]_{F_{\mathunderscore}}&&& \option(\A')^{\NN \times \Q'} \ar@{..>}[u]_{{\delta_{\B'}}} 
  }
  }
  \caption{The multivalued function $F$ realizes the partial function $f$.
    If $M$ is a machine for $F$ then $F_M$ tightens $F$ and therefore also realizes $f$.
    The function $\delta_{\B'} \circ M$ chooses through $F_M$ where $\delta_{\B'}$ is the precomplete representation of Baire space described in Section~\ref{sec: precompleteness}.
    We use dotted arrows to emphasize that $F$ is multivalued while $M$ and $\delta_{\B'}$ are singlevalued.
    Replacing $\delta_{\B'}$ with a multi-representation recovers $F_M$ instead of a choice function. } 
    \label{fig: machine}
\end{figure}
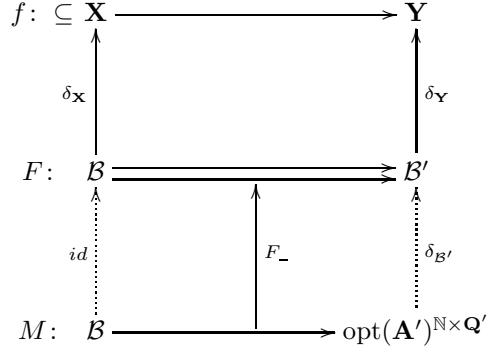
\subsection{Inversion in the rational representation}\label{sec: inversion}
Consider $x \mapsto \nicefrac 1x$ as a partial function on the represented space $\RR_\QQ$ from Example \ref{ex: rational reals}.
The function is partial as it is undefined in $0$.
We define a function $M\colon \B \to \option \QQ^{\NN \times \QQ}$ of which we claim that $F_M \colon \B \mto \B$ is a realizer of inversion:
\[ M(\varphi)(n, \varepsilon) :=
\begin{cases}
  \some\displaystyle\frac1{\varphi(\min\{\delta,\varepsilon\delta^2\}/2)} &\text{if } \delta := \abs{\varphi(2^{-n})} - 2^{-n} > 0\\
  \none & \text{otherwise.}

\end{cases} \]
Unfolding of definitions reveals that we have to prove that for all $x \neq 0$
\[ \delta_{\RR_\QQ}^{-1}(x) \subseteq \dom(F_M) \land \forall \varphi \in \delta_{\RR_\QQ}^{-1}(x)\colon F_M(\varphi) \subseteq \delta_{\RR_\QQ}^{-1}(\nicefrac 1x). \]
This should be understood as two statements:
Firstly that the domain of $F_M$ includes all names of real numbers from the domain of the inversion function, and secondly that it only returns correct values.

To prove the first of these statements, let $\varphi$ be a name of some $x\neq 0$.
It suffices to pick $n$ larger than $\log_2(|\nicefrac1x|)$ to avoid the second case and thus the domain of $F_M$ is big enough.
To check the second condition, i.e.\ that $F_M$ only returns correct values, let $\varphi$ be a name of $x$.
Here we use the property discussed below Example \ref{ex: booleans}, namely that the rational representation admits a separation into correct and incorrect answers.
Thus it is sufficient to check for each $\varepsilon>0$ that $M(\varphi)(n, \varepsilon) = \some r$ implies that $|r - \nicefrac 1x| \leq \varepsilon$.
If the assumption of this implication is true, then we have $\delta := |\varphi(2^{-n})| - 2^{-n} > 0$ and we know the value of $r$.
First note that $|x| \geq \delta$ as can be seen using the inverse triangle inequality and that $\delta$ is positive.
This, together with another application of the inverse triangle inequality, leads to
\[ |\varphi(\min\{\delta, \varepsilon \delta^2\}/2)| \geq \Big||x| - \frac{\min\{\delta, \varepsilon \delta^2\}}2\Big| \geq \frac \delta 2 \]
and allows us to conclude that
\begin{align*}
  \Big|\frac 1 {\varphi(\min\{\delta, \varepsilon \delta^2\}/2)} - \frac 1 x\Big| & = \frac{|\varphi(\min\{\delta, \varepsilon \delta ^2\}/2) - x|}{|\varphi(\min\{\delta, \varepsilon \delta ^2\}/2)x|} \\
  & \leq \frac{\min \{\delta, \varepsilon \delta^2\}}{\delta^2} \leq \varepsilon.
\end{align*}
As the left hand side is the value of $r$ this proves the correctness of return values.

The function $M$ is computable as all operation it uses on the rational numbers are computable.
Note that $F_M$ is properly multivalued.
In general, computability should imply continuity, but this does not make sense here as continuity only makes sense for single-valued functions.
This can be resolved by removing the multivaluedness of $F_M$ via picking its value on the smallest $n$ for which it returns something.
The function obtained in this way is a realizer since realizability is preserved under tightening.
As searching is a computable operation, this realizer is moreover computable, which is reflected in the fact that it can be defined as a dependently typed function in \Coq{} from the definition of $M$ as above.

\section{Machines as names of functions}\label{sec: continuous machines}

It is true that for every partial operator $F$ there exists some machine $M$ such that $F_M$ extends $F$.
This means that we can understand $M$ as a description of $F$ in a similar way to how representations work.
Nevertheless, as the candidates for question and answer sets are full function spaces and thus uncountable, this does not formally define a representation.
Access to $M$ alone is an inconvenient set of information in the sense that it is difficult to maintain through operations.
For instance, given machines for each of two operators $F$ and $G$, it can not be easily found for the operator $F\circ G$.
This is because $G(\varphi)$ and thus the input for $F$ can only be approximated from access to a machine for $G$.
Only when restricting to continuous operators, can one hope to succeed by extending some finite sub-function in an arbitrary way.
To guarantee that this does not interfere with the correctness of the return values, one needs explicit information about the continuity of $F$.
A set of such information that is often used in constructive analysis is a modulus function.

Fix sets $\Q$, $\A$, $\Q'$ and $\A'$, recall that $\Q^*$ denotes the finite lists of elements from $\Q$ and use the abbreviations $\B := \A^\Q$ and $\B':= \A'^{\Q'}$.
A function $\mu\colon {\subseteq \B} \to \Q^{*\Q'}$ is called a \demph{modulus} of an operator $F\colon {\subseteq \B} \to \B'$ if it is a Skolem function of the continuity statement from Section~\ref{sec: continuity} in the sense that for all $\varphi, \psi \in \dom(F)$ and $q' \in \Q'$
\begin{equation}\label{eq: continuity implication}
\varphi|_{\mu(\varphi)(q')} = \psi|_{\mu(\varphi)(q')} \implies F(\varphi)(q') = F(\psi)(q').
\end{equation}
In particular $\dom(F) \subseteq \dom(\mu)$ as otherwise the premise of the implication does not make sense.
A modulus $\mu$ may itself be considered an operator and the type of a modulus of $\mu$ coincides with the type of $\mu$ itself.
It thus makes sense to call a modulus \demph{self-modulating} if it is its own modulus.

\begin{proposition}\label{resu: self-modulus}
  Any continuous partial operator has a self-modulating modulus of continuity.
\end{proposition}
\begin{proof}
  Fix an enumeration of $\Q$.
  Let $\mu$ be the function that returns the minimal initial segment with respect to this enumeration such that the implication (\ref{eq: continuity implication}) is fulfilled.
  Since $F$ is continuous, $\mu$ is well defined.
  It is a modulus by definition and it can be checked that it is also self-modulating.
\end{proof}
The above proof is non-constructive: Checking whether the implication from equation \ref{eq: continuity implication} holds is not necessarily decidable and the argument that the modulus is well-defined thus relies on the law of excluded middle.
We suspect that the use of a classical background theory is indispensable as there are well known obstructions to constructively proving continuity statements about moduli \cite{MR966421, ESCARDO2016770}.
A direct argument that the idea behind the proof above is inherently non-constructive is that there exist computable operators whose minimal modulus is not computable.
When specialized to such an operator the proof asserts the existence of a non-computable function from computable input data and can thus not be made constructive without modifying its core idea of using the minimal modulus.

The proof that we stated for proposition \ref{resu: self-modulus} was chosen for being illustrative and short.
As stated it is a lot more non-constructive than necessary and in the formal development we proceeded differently in attempt to minimize the strength of the choice principles that get involved.
Note that from a set-theoretic point of view, the above does not require resorting to any kind of choice principle.
In set theory a function is a relation that uniquely specifies a return value for every input.
Choice principles typically select one of many possible return values.
In our case the uniqueness is guaranteed by the minimality condition and a selection is not necessary.
For understanding in how far the modified proof still improves over the one above in respect of the use of choice principles requires some background in type-theory.
For this reason we refrain from stating the details here and only briefly sketch the ideas.

From a type-theoretic point of view the proof of proposition \ref{resu: self-modulus} only ever talks about a specification of a unique functional modulus and not about the modulus function itself.
It specifies a property that the return-value of the modulus should have, and argues that such a return value always exists and is unique but does not say how to construct it.
In other words, the proof does not give a function definition of $\mu$ in the sense of type theory, as such a definition carries algorithmic information.
The process of filling in the missing piece and going from a relational specification to a function definition is referred to as ``functional relation reification'' in type theory.
Ideally, the procedure involves filling an abstract idea with concrete meaning.
As we already discussed, in the present case there is reason to doubt the feasibility of providing full algorithmic information.
It is possible to close the remaining gap axiomatically by use of a so-called ``no choice'' principle.
Of course, in doing so one should attempt to minimize the strength of the inefficient principle one uses.

So let us briefly discuss how inefficient such principles typically are.
A ``no choice'' principle states that for every specification that uniquely points to a return value there exists a function that fulfills this specification.
The strength of no choice principles can change considerably depending on the input type of the function to be constructed.
We will refer to the no choice principle for all functions with inputs of a certain type as the no choice principle \demph{over} that type and are mostly interested in those over the natural numbers and Baire-space.
While in constructive mathematics countable choice is often considered valid, in our classical setting already the countable version of a ``no choice'' principle is distinctively inefficient:
A typical application of this principle is to assert the existence of a non-computable function such as the characteristic function of the halting set as a subset of the natural numbers.
Uses of no choice principles over Baire space, by contrast, are typically needed when the function whose existence should be assured is not only incomputable but even discontinuous.
For instance consider the characteristic function of a one point subset of Baire-space, this function is discontinuous and its existence can be proven using the ``no choice'' principle over Baire-space.
The use of the ``no choice'' principle over Baire-space is particularly worth avoiding as it comes with a direct contradiction to other principles one may want to assume such as the uniform continuity principle.

Let us turn back to Proposition \ref{resu: self-modulus} and understand it to state that every continuous partial operator has a self-modulating modulus function in the type theoretical sense.
The proof does only talk about the relational specification, but the remaining gap can straight-forwardly be closed by use of the no choice principle over Baire space.
One may ask whether the full strength of this principle is necessary in this case, in particular as the constructed function is continuous.
Indeed, a choice principle over the natural numbers is sufficient for proving the proposition:
One may first use the countable choice principle to find a dense sequence $\varphi_n$ in the domain of the operator.
Then use it again to choose a minimal certificate for each combination of one of the countably many $q' \in \Q'$ and an element of this sequence.
From this data, the values of a self-modulating modulus can explicitly be constructed through an iteration:
Given some inputs $\varphi$ and $q'$, starting from the empty list iterate the process of first finding an element of the dense sequence that coincides with $\varphi$ on the given list and then updating the list to contain the minimal certificate associated with that element.
This iteration stabilizes detectably and the function that returns the list it stabilizes on can be proven a self-modulating modulus.
However, this does not produce the values of the minimal modulus but in general returns an over approximation.
Moreover, we used a proper choice principle and not only the ``no choice'' principle as the selection of a dense sequence in the domain involves selecting from a possibly uncountable number of candidate values for each natural number.

Let us close this discussion with two remarks.
The first is that the above argument still seems inherently inefficient.
This is because there exist algorithms operating on Baire-space whose natural domain of definition is non-empty but does not contain any computable elements \cite{kleene1953recursive}.
In such a case the sequence $\varphi_n$ is necessarily incomputable and thus it should not be expected to be constructible from the operator in general.
However, the typical examples of operators computed by such algorithms have total computable extensions (computed by a distinct algorithm) so that this argument cannot straight-forwardly be made rigorous.
At this point we should also make clear that the above informal description glossed over some of the details concerning partiality.
The second and final remark is that typically ``no choice'' principles in type theory are formulated such that they assert the propositional existence of the function.
The function whose existence is asserted can not directly be used definitionally at the type level.
This is important as it means that a machinery such as \Coq{}'s code extraction that removes non-computational parts of proofs will remove the inefficiencies and still return executable code.
Of course, the correctness of the executed code is now only provided under the condition that the assumptions were non-contradictory.
Direct executability via reduction in a type-theory based proof-assistant, on the other hand, is blocked by the use of such principles.
As we work over a classical background theory this is an expected outcome anyways.

\begin{definition}
Call a pair $(M, \mu)$ a \demph{continuous machine} if $M$ is of type  $\B \to \option \A'^{\NN \times \Q'}$ and $\mu$ is a self-modulating modulus of $M$.
Say that a continuous machine $(M, \mu)$ implements an operator $F \colon{\subseteq \B} \to \B'$ if $F_M$ tightens $F$.
\end{definition}

Let us emphasize that the function $\mu$ above is a modulus of continuity of the machine $M$ itself and not of a potential operator $F$ that it computes.
In particular, just like $M$ itself, the modulus $\mu$ is always a total function.
Proposition~\ref{resu: modulus} below implies that from $\mu$ one can obtain a modulus of continuity of any operator that is tightened by $F_M$.
However, it is not difficult to construct a discontinuous $M$ such that $F_M$ is a continuous partial function and thus a modulus of $F_M$ is not enough information to recover one of $M$.
Thus, a modulus of the computed operator should be considered to provide strictly less information than $\mu$.

Recall that earlier we started from a computable $F$ and discussed how to construct an appropriate $M$ from an oracle machine that computes $F$ (see Fig.~\ref{fig: FM}).
A computable, self-modulating modulus $\mu$ for $M$ can be readily read off the oracle machine by following the queries that the machine writes to its oracle tape.
The resulting pair $(M, \mu)$ is a continuous machine that implements $F$.
More generally, every continuous operator can be implemented by a continuous machine.

\begin{proposition}\label{resu: machines}
  If $F\colon {\subseteq \B} \to \B'$ is continuous then there exists some continuous machine that implements it.
\end{proposition}
\begin{proof}
  Let $d\colon \seq(\Q \times \A) \to \option(\B)$ be a function that, if the input list is the graph of a finite function $\phi$, returns some $\varphi \in \dom(F)$ such that $\phi = \varphi|_{\dom(\phi)}$ if such a $\varphi$ exists and otherwise returns $\none$.
  Let $\mu$ a self-modulating modulus of $F$ that exists by Proposition~\ref{resu: self-modulus} and $(q_n)_{n\in\NN}$ an enumeration of $\Q$.
  Let $M$ be given by
  \[ M(\varphi)(n, q') := 
  \begin{cases}
    \some (F(\varphi')(q'))  & \text{if } d(\varphi|_{(q_1, \ldots, q_n)}) = \some \varphi' \\
    & \text{and } \mu(\varphi')(q') \subseteq (q_1, \ldots, q_n) \\
    \none & \text{otherwise.}
  \end{cases}
  \]
  If $M$ returns something, the return value is correct because $\mu$ is self-modulating.
  On the other hand, whenever $\varphi\in\dom(F)$, there exists some $n$ such that $\mu(\varphi)(q') \subseteq (q_1, \ldots, q_n)$ and for this $n$ the machine reproduces the value of $F$.
  Clearly, the values of $M$ depend only on the values of $\varphi$ on $(q_1, \ldots, q_n)$, where $n$ is such that $\mu(\varphi')(q') \subseteq (q_1, \ldots, q_n)$.
  Just returning $(q_1, \ldots, q_n)$ is a modulus of $M$ that is independent of $\varphi$ and thus self-modulating.
\end{proof}
This proof uses classical reasoning and countable choice for the construction of the function $d$ and for picking an enumeration of $\Q$.
There exist algorithms whose natural domain of convergence is non-empty but does not contain any computable elements and it thus seems unlikely that an appropriate $d$ can be explicitly constructed even with access to an algorithm for computing $F$.
Just as before, the proof listed here is simpler than the formal version, where we adaptively increase the size of the initial segment depending on the values of the modulus.
This decreases the effort needed for a successful evaluation and while it introduces a dependency of the modulus on its functional input, it also forces its return values to eventually stabilize with increasing effort.
It even has a stronger property that, as introduced in the discussion of monotone machines below, it ``terminates with $M$''.

Just like it is possible to reconstruct the values of $F$ from $M$, a modulus for $F$ can be reconstructed using the additional information that a continuous machine implementing it provides.
We omit the somewhat straight-forward proof.
\begin{proposition}\label{resu: modulus}
  A machine that computes a modulus for $F\colon{\subseteq \B} \to \B$ can be obtained in a fully uniform way from a continuous machine that implements $F$.
  The construction can be done in such a way that it preserves being self-modulating.
\end{proposition}
Here and in the following results by ``fully uniformly'' we mean that the transformation can be defined in a fragment of \Coq{}'s type theory small enough to not go beyond definability in system T when all the question and answer types are the natural numbers.
Adding a self-modulating modulus $\mu$ to a machine completes the set of information about $F$ in the sense that a continuous machine implementing a realizer of some function between represented spaces contains exactly the amount of information that one would expect to be specified about such a function in computable analysis.
To understand this in more detail let us first recall how computable analysis treats spaces of functions.

\subsection{Function spaces and continuous machines}\label{sec: function spaces}
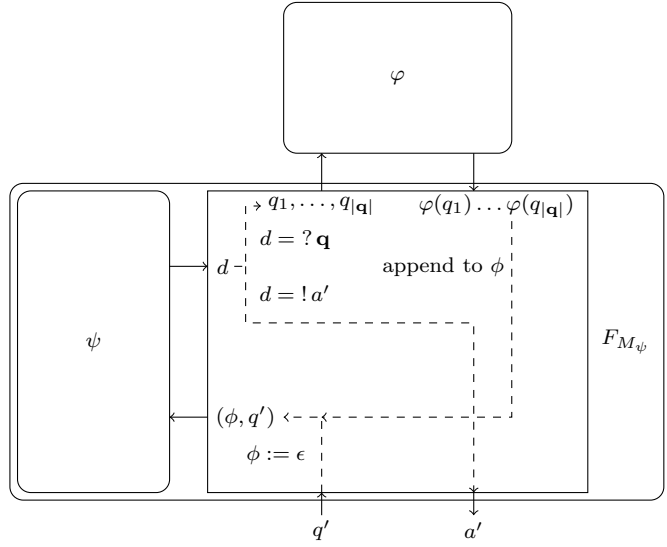
\begin{figure}
  \centering
  \begin{tikzpicture}
   \tikzstyle{every node}=[font=\footnotesize]  
   \draw (0,-1) rectangle (5,3);
   \draw[rounded corners = 5pt] (-2.6,-1.1) rectangle (6,3.1);
   \node at (5.5,1) {$F_{M_{\psi}}$};
   \draw[->,dashed] (1.5,-1) -- (1.5,0) -- (1,0);
   \node at (0.9,-.5) {$\phi := \epsilon$};

   \draw[->] (0,0) -- (-.5,0);
   \node at (.5,0) {$(\phi, q')$};
   
   \draw[rounded corners = 5pt] (-2.5,-1) rectangle (-.5,3);
   \node at (-1.5,1.0) {$\psi$};
   \draw[->] (-.5,2) -- (0,2);
   \node at (.2,2) {$d$};
   \draw[->,dashed] (.35,2) -- (.5,2)--(.5,2.8)--(0.7,2.8);
   \node at (1.15,2.35) {$d = \ask \str q$};
   \node at (1.15,1.65) {$d = \answer a'$};
   \draw[->,dashed] (.5,2) -- (.5,1.25) -- (3.5,1.25) -- (3.5,-1);
   
   \node at (1.5,2.8) {$q_1, \dots, q_{\size {\str q}}$};
   \draw[->] (1.5,3) -- (1.5,3.5);
   \draw[rounded corners = 5pt] (1,3.5) rectangle (4,5.5);
   \node at (2.5,4.5) {$\varphi$};
   \draw[->] (3.5,3.5) -- (3.5,3);
   \node at (3.8,2.8) {$\varphi(q_1)\ldots\varphi(q_{\size{\str{q}}})$};
   \draw[->,dashed] (4,2.6) -- (4,0) -- (1.5,0);
   \node at (3.1,2.0) {append to $\phi$};
   
   \node at (1.5,-1.5) {$q'$};
   \draw[->] (1.5,-1.3) -- (1.5,-1);
   \draw[->] (3.5,-1) -- (3.5,-1.3);
   \node at (3.5,-1.5) {$a'$};
  \end{tikzpicture}
  \caption{
    Some $\psi$ is a name of a function $f: \X \to \X'$ iff $F_{M_\psi}$ realizes $f$.
    The box with pointy corners represents a realistically implementable algorithm while $\varphi$ and $\psi$ may be computable or non-computable and are therefore depicted with rounded corners.
    Whenever $\psi$ is computable and has pointy corners, also $F_{M_\psi}$ will be computable and can be depicted with pointy corners.
  }\label{fig:U}
\end{figure}
In this part we make an additional assumption about the question type $\Q$, namely that it features decidable equality.
We do so to make it reasonable to encode finite functions as lists of input/output pairs i.e. over $\seq(\Q \times \A)$.
The decidable equality on $\Q$ is needed to make evaluation and checking for inclusion of a finite list in the domain of the finite function definable with respect to this encoding.
We hide this encoding and it only surfaces in our use of $|\phi|$ for the size of a finite function $\phi$, which we understand to denote the number of elements of the list and not the size of the domain of $\phi$ which may be smaller due to argument value pairs being repeated.
Also recall that in \Incone{} each question type of a represented space comes with a default question $q_d\in\Q$.
For convenience this section sometimes also assumes a default answer to be available.
The use of default answers can generally be avoided by first asking the default question and using the returned answer as default answer.
The assumption of decidability of equality on $\Q$, on the other hand, is not only for convenience, but essential for the argument.

Fix some represented spaces $\X$ and $\X'$ whose spaces of names are $\B = \A^\Q$ and $\B' = \A'^{\Q'}$ respectively.
In the following use $?$ as notation for the right inclusion into $\A' + \seq \Q$ and $!$ for the left inclusion, i.e.\ a question mark for a list of questions and an exclamation mark for an answer.
For a fixed function $\psi\colon \seq(\Q \times \A) \times \Q' \to \A' + \seq \Q$ and fixed $\varphi\in \B$ and $q' \in \Q'$ inductively define a sequence of finite functions $\phi_n \in \seq(\Q \times \A)$ by $\phi_0:= \epsilon$ and
\[ \phi_{n+1}:=
\begin{cases}
  \phi_n \cat {\varphi|_K} & \text{if } \psi(\phi_n,q') = {?K} \\
  \phi_n & \text{otherwise.}
\end{cases}\]
From this sequence define a machine $M_\psi$ as follows:
\[ M_\psi(\varphi)(n,q') :=
\begin{cases}
  \some a' & \text{if } \psi(\phi_n, q') = {!a'} \\
  \none & \text{otherwise.}
\end{cases}
\]
For illustration see Figure~\ref{fig:U}.

The function-space representation $\delta_{\X'^\X}$ assigns $\psi$ as name to a function $f\colon \X \to \X'$ if and only if $F_{M_\psi}$ realizes $f$.
This means that the space of names of functions is given by $\B_{\X'^\X}= (\A' + \seq(\Q))^{\seq (\Q \times \A) \times \Q'}$ and in particular the questions and answers of the function-space representation are countable and the standard question is the empty list paired with the standard question of $\X'$.
While our presentation is quite different, the central idea coincides with that behind Weihrauch's $\eta$ \cite{weihrauch_computable_2000}.
The function-space representation is precomplete but we do not go into detail about this here.
Straightforward computations show that $\mu_\psi(\varphi)(n, q') := \dom(\phi_n)$ is a self-modulating modulus of $M_\psi$ and thus $(M_\psi, \mu_\psi)$ is a continuous machine and that $F_{M_\psi}$ is single-valued and therefore also continuous by Proposition~\ref{resu: modulus}.
It is also true that every continuous function is given a name, thus the set underlying the represented space $\X^{\X'}$ are exactly the continuous functions.
Actually, the latter can be understood as a consequence of Proposition~\ref{resu: machines} together with the argument below.
Call $\psi$ an \demph{associate} of $F\colon\B \mto \B'$ if $F_{M_\psi}$ tightens $F$.
Then $\psi$ is a name of $f$ if and only if it is an associate of a realizer of $f$ and any associate of $F$ can be used to obtain a continuous machine that implements $F$.
To prove the converse of this statement it is sufficient to show that from a continuous machine $(M, \mu)$ one can obtain an associate $\psi_{M,\mu}$ of $F_M$.

To get an intuition for what an associate of $F_M$ should be doing, first consider the case where not only a continuous machine but an actual oracle machine is available.
The main obstacle in this case is that the associate is required to be a total function and divergences of the oracle machine need to be taken care of.
Define an associate $\psi$ of the operator computed by an oracle machine as follows:
given a finite function $\phi$ and some question $q'$ run the oracle machine for at most $\abs{\phi}$ steps while looking up the answers to the questions that the oracle machine asks in the finite function.
If the lookup fails for a question $q\in \Q$, then return $?(q)$.
If all lookups are successful and the machine terminates with return value $a'$ return $!a'$.
In case that $\abs{\phi}$ steps are exceeded without either happening, return $?(q_d)$ where $q_d$ is the default question of $\X$.
If the oracle machine comes to an end on input $q'$ and oracle $\varphi$, then this value is eventually reproduced by $M_\psi$ as $|\phi_n|$ grows with the effort $n$ until it hits the number of steps the oracle machine needs to conclude.

As the next step, let us discuss how to supplement full inspection capabilities into what $M$ does with access to a self-modulating modulus.
For illustration, we consider only the special case where $F$ is given as a total function together with a self-modulating modulus $\mu$.
That is, we drop the effort parameter and remark that this simplification is partially justified by the last paragraph, which argues that divergences can be taken care of.
Thus an associate of $F_M$ should, for fixed inputs $\varphi$, $q'$, attempt to get hold of $\mu(\varphi)(q')$, as this is the set of questions, the answers to which its final answer should depend on.
However, the associate only has access to a finite sub-function $\phi$ of $\varphi$.
Let us assume that a default answer $a_d\in\A$ is available and set
\[ \varphi_\phi(q) :=
\begin{cases}
  \phi(q) & \text{if } q \in \dom(\phi) \\
  a_d & \text{otherwise.}
\end{cases} \]
The associate may on input of $\phi$ and $q'$ use $\varphi_\phi$ as a replacement for $\varphi$.
However, $\varphi$ and $\varphi_\phi$ can only be expected to coincide on $\dom(\phi)$ and there does not seem to be a reason why $\mu(\varphi)(q')$ should have anything to do with $\mu(\varphi_\phi)(q')$.
This is where the property of being self-modulating comes in:
the values of the modulus coincide whenever $\mu(\varphi_\phi)(q') \subseteq \dom(\phi)$, and this is a condition that can be checked by the associate.
With this motivation, let the associate on input of $\phi$ and $q'$ check whether $\mu(\varphi_\phi)(q') \subseteq \dom(\phi)$ and if this test fails ask for the difference, i.e.\ return $?(\mu(\varphi_\phi)(q') \setminus \dom(\phi))$.
If the test is successful, then $\mu(\varphi_\phi)(q') =\mu(\varphi)(q')$ and the associate can safely return $!F(\varphi_\phi)(q')$ as $\mu$ is a modulus of $F$.

From the construction it should be clear that any value that is returned is correct.
However, as the modulus is evaluated on functional inputs that are different from the actual input in relevant places, an argument is needed to see that the iteration is finite.
Recall the sequence $\phi_n$ of finite functions that was used to define associateship and first argue that the sequence $\varphi_n := \varphi_{\phi_n}$ converges to some limit $\psi \in \B_\X$.
This is because for some fixed $q \in \Q$ either there exists some $n$ such that $q \in \dom(\phi_n)$, in which case $\varphi_k(q) = \varphi(q)$ for all $k$ bigger than $n$, or there does not exist such an $n$ and $\varphi_k(q) = a_d$ for all $k$.
From this it should be clear how the values of $\psi$ have to be picked.
As $\mu$ is self-modulating, it is continuous, and since all question and answer types are countable, it is also sequentially continuous.
Thus $\mu(\varphi_n)$ converges to $\mu(\psi)$ and thus there exists some $k$ such that $\mu(\varphi_m)(q') = \mu(\psi)(q')$ for any $m$ bigger or equal $k$.
In particular
\begin{align*}
  \mu(\varphi_{k+1})(q') & = \mu(\psi)(q') = \mu(\varphi_k)(q') \\
  & \subseteq \dom(\phi_k) \cat \mu(\varphi_k)(q')\setminus \dom(\phi_k) \\
  & = \dom(\phi_{k+1}),
\end{align*}
and $k+1$ is a sufficiently large effort to lead the evaluation of the associate to return a value.

\begin{theorem}\label{resu: main}
  There exists a fully uniform way to construct from a continuous machine $(M, \mu)$ and default elements $q_d \in\Q$ and $a_d\in \A$ an associate of $F_M$.
\end{theorem}
\begin{proof}
  The core ideas of the proof should be clear from the informal arguments above.
  The important points are that firstly the unbounded search for an effort big enough for $M$ can be moved to the search of a large enough effort in $M_\psi$, and secondly one can iterate the modulus to get hold of information about what values $M$ actually needs to know the values of $\varphi$ for.
  Combining these ideas, one can prove that the following associate is correct.
  Set $\psi_{M,\mu}(\phi, q') := !a'$ if there exists an $m \leq |\phi|$ such that $M(\varphi_\phi)(m, q') = \some b'$ and $\forall n \leq m \colon \mu(\varphi_\phi)(n, q') \subseteq \dom(\phi)$ and $a'$ is the $b'$ for the smallest such $m$.
  If there exists an $m \leq |\phi|$ such that $\mu(\varphi_\phi)(m,q') \not\subseteq \dom(\phi)$ and $\forall n \leq m\colon M(\varphi_\phi)(n,q') = \none$, then set $\psi_{M,\mu}(\phi, q') := 
?\mu(\varphi_\phi)(m', q')\setminus \dom(\phi)$ where $m'$ is the least such $m$.
  In any other case return $?(q_d)$.
  Figure \ref{fig:assoc} illustrates the main idea of the proof.
  We point the reader to the formal development for the details.
\end{proof}

The associate uses the extension $\varphi_\phi$ of a finite function $\phi$ and thus assumes availability of a default element $a_d\in \A$.
This assumption could be dropped if we use the first return value listed in $\phi$ instead of $a_d$  and return $?(q_d)$ if $\phi$ is empty.

\begin{figure}
  \centering
  \begin{tikzpicture}
    \draw[rounded corners = 5pt] (-3.5,-3) rectangle (6.2,4);
    \draw (8,-3) -- (6.75,-3) -- (6.75,4) -- (8,4);
    \node at (7.5,.5) {Fig. \ref{fig:U}};
    \node (S) at (6.2,-1.5) {};
    \draw[->] ($(S)+(0,4)$) -- ++(0.5,0);
    \draw[->] ($(S)+(0.55,0)$) -- ++(-0.55,0) node[left]{$({\color{magenta}{\phi}},{\color{blue}{q'}})$} ;

    \draw[->,color=magenta] ($(S)-(0.8,0.2)$) -- ++(0,-1.2) -- ++(-8.4,0) -- ++(0,6.1) node (P0){} -- ++(5.6,0) -- ++(0,-0.2);
    \draw[->,color=magenta] ($(P0)+(1.6,0)$) -- ++(0,-0.2);
    \draw[->,color=blue] ($(S)-(0.4,0.2)$) -- ++(0,-1) node (P1){} -- ++(-7.7,0) -- ++(0,3.2);
    \draw[->,color=blue] ($(P1)+(-3.7,0)$) -- ++(0,3.2);

    \begin{scope}[shift={(4,0.5)}]
      \node at (-1.5,2.4) {$\varphi_{\color{magenta}{\phi}}$};
      \draw[->] (-1.5,2.2) -- (-1.5,2);
      \draw[rounded corners = 5pt] (-2.5,2) rectangle (-.5,1);
      \node at (-1.5,1.5) {$\mu$};
      \draw[->] (-1,1) -- (-1,0.5) node[left, midway]{};
      \draw[->] (-2,0.5) node (Q1)[below]{\footnotesize ($i,{\color{blue}{q'}})$} -- (-2,1) ;
      \node (K) at (-0.3,0.3) {\footnotesize $\mathbf q \overset{?}{\subseteq} \dom {\color{magenta}{\phi}}$};
    \end{scope}

    \begin{scope}[shift={(0,0.5)}]
      \node at (-1.5,2.4) {$\varphi_{{\color{magenta}{\phi}}}$};
      \draw[->] (-1.5,2.2) -- (-1.5,2);
      \draw[rounded corners = 5pt] (-2.5,2) rectangle (-.5,1);
      \node at (-1.5,1.5) {$M$};
      \draw[->] (-1,1) -- (-1,0.5) node[left, midway]{};
      \draw[->] (-2,0.5) node (Q2)[below]{\footnotesize ($i,{\color{blue}{q'}})$} -- (-2,1) ;
      \node (A) at (-0.5,0.3) {\footnotesize $o \overset{?}{=} \none$};
    \end{scope}

    \draw[dashed, ->] ($(K)+(-0.1,0.1)$) -- ++(0.7,0.5) node[midway,above,sloped]{\footnotesize \xmark} node[right] (Kx) {};
    \draw[dashed, ->] ($(K)+(-0.25,-0.4)$) -- ++(0,-0.8) node[midway,left]{\footnotesize \cmark} node (Ko) {};

    \draw[dashed, ->] (Kx) -- ++(1.6,0) node[above left]{\footnotesize $\rightsquigarrow ? (\mathbf q \setminus \dom {\color{magenta}}{\phi})$} node (Ke) {};;

    \draw[dashed, ->] ($(A)+(-0.1,0.1)$) -- ++(0.7,0.5) node[midway,above,sloped]{\footnotesize \xmark} node[right] (Ax) {};
    \draw[dashed, ->] ($(A)+(-0.3,-0.3)$) -- ++(0,-0.8) node[midway,left]{\footnotesize \cmark} node [left] (Ao) {};
    
    \draw[dashed,->] (Ax) -- ++(0,2.1) -- ++(4,0) node[midway,above]{\footnotesize $o = \some(a')$} -- ++(1.8,0) node[midway,above,sloped]{\footnotesize $\rightsquigarrow !a'$} node (Ae) {};

    \node (I) at ($(A)+(-0.15,-1.4)$) {\footnotesize $i \overset{?}{<} \size {{\color{magenta}{\phi}}}$};
    \draw[dashed,->] ($(I)+(0,0.1)$) -- ++(0.7,0.5)node[midway,above,sloped]{\footnotesize \xmark} node[right] (Ix) {};
    \draw[dashed,->] ($(I)+(-0.1,-0.3)$) -- ++(0,-0.8) node[midway,left]{\footnotesize \cmark} node (Io) {};

    \draw[dashed,->] (Ix) -- ++(.8,0) -- ++(0,3.4) -- ++(3.1,0) -- ++(0,-.9) -- ++ (1.9,0) node [midway, above] {\footnotesize $\rightsquigarrow ?(q_d)$} node (Ie) {};

    \draw[dashed] (Ke) -- (Ae);
    \draw[dashed] (Ie) -- ++(0.2,0);
    
    \draw[dashed, ->] (Ko) -- ++(0,-2) -- ++(-5.6,0) -- ++(0,2.9) node (P1) [midway] {};

    \draw[dashed, ->] (Io) -- ++(0,-0.3) -- node[midway,above]{\footnotesize increase $i$} ++(2.55, 0) -- ++(0, 0.45);
 
    \draw[->,dashed] ($(S) + (-1.2,0)$) -- ++(-0.5,0) node (It)[left] {\footnotesize $i := 0$};
    \draw[->,dashed] (It) -- ++(-2.2,0) -- ++(0,2);
  \end{tikzpicture}
  \caption{
    Constructing an associate $\Psi_{M,\mu}$ from a continuous machine $(M, \mu)$.
    Here, $\varphi_\phi$ is the total function that extends $\phi$ with a default value.
    If $M$ and $\mu$ come with algorithms to compute them, we obtain an algorithm for $\psi_{M,\mu}$.
    In the picture this could be illustrated by making all rounded corners pointy.
 }
  \label{fig:assoc}
\end{figure}
\subsection{Continuous machines and monotonicity}

Continuous machines and associates theoretically contain the same information about a continuous operator.
However, in practice continuous machines are vastly superior to associates when the task is to directly implement an algorithm and formally prove it correct.
The skeptical reader may revisit the example of inversion on the rational reals from Section~\ref{sec: inversion}, supplement a self-modulating modulus and extract an associate.
Totality of the associate and encoding finite functions by lists introduces irrelevant default values making proofs of correctness tedious.
The translation from continuous machines takes part of the burden off the user.

As the concept of an associate is linked to partial combinatory algebras (compare for instance \cite{van2011partial}), many operations on continuous operators are in principle implementable for associates and thus also for continuous machines.
For implementation of operations on continuous operators, both working with associates and working with machines is unhandy but for somewhat different reasons.
While associates are difficult to construct, a continuous machine as input makes some important information not readily available.
One way of reflecting the difference in rigidity of the concepts is to translate back and forth between them.
While any continuous machine can be translated to an associate, the machines that are obtained from an associate have very special properties some of which can be maintained separately.

A property of continuous machines that can be propagated with relatively low effort and vastly simplifies implementation of operations such as the composition of operators is monotonicity in the following sense:
Call a machine $M$ \demph{monotone} if $M(\varphi)(n, q') = \some a'$ implies that for any $m \geq n$ it holds that $M(\varphi)(m, q') = \some a'$.
Call a continuous machine $(M, \mu)$ a \demph{monotone machine} if $M$ is monotone and $\mu$ \demph{terminates with} $M$ in the sense that once $M$ returns a value on some inputs, further increasing the effort on the same inputs does not lead $\mu$ to return bigger lists anymore.

For a monotone $M$, the corresponding $F_M$ is single-valued.
The machine we used to implement inversion in Section \ref{sec: inversion} is not monotone.
Any continuous machine constructed from an oracle machine as outlined in the introduction of this section and also those constructed from associates as outlined in the previous subsection are monotone.
Thus, if the equality on $\Q$ is decidable, translating from a continuous machine to an associate and back allows to force monotonicity.
A more direct method does not require any additional assumptions about question and answer sets.

\begin{proposition}
  From a continuous machine $(M, \mu)$ a monotone machine that implements a choice function of $F_M$ can be obtained.
  This construction can be done fully uniformly.
\end{proposition}

\begin{proof}
  Consider the monotone machine $\usefirst M$ (for ``use first'') defined as follows:
  on input of $\varphi$, $n$ and $q'$ search for the smallest $m \leq n$ such that a return-value is produced and return this value, if no such $m$ exists return $\none$.
  As $\usefirst M$ is monotone, $F_{\usefirst M}$ is a partial function and it respects the interpretation of $M$ in the sense that $F_{\usefirst M}$ is a choice function for the multivalued function $F_M$.

  A version $\usefirst \mu$ of the modulus such that $(\usefirst M, \usefirst \mu)$ is a monotone machine can be defined by
  \[ \usefirst \mu(\varphi)(n, q') := \bigcup_{\{i | i \leq n \land \forall j < i\colon M(\varphi)(j, q') = \none\}} \mu(\varphi)(i, q'). \]
  We omit the straight forward computation that this modulus is appropriate.
\end{proof}

The modulus takes a union over all previous values, which leads to an undesirable overestimation.
As a consequence, the modulus is monotone in the sense that the lists it returns grow with increasing effort and this property, while it can be a useful, is not required for the modulus of a monotone machine.
One may be tempted to modify the construction by omitting the values of the modulus where $M$ returns $\none$.
Unfortunately, the function obtained in this way is in general neither a modulus of $\usefirst M$ nor self-modulating.
\begin{example}[Modulus-failure]
  Set $\Q := \one =: \Q'$ and $\A := \BB =: \A'$ and consider the machine
  \[ M(\varphi)(n, \star) :=
  \begin{cases}
    \none & \text{if } n = 0 \text{ and } \varphi(\star) = \false \\
    \some \false & \text{if } n = 0 \\
    \some \true& \text{otherwise,} \end{cases} \]
  together with the self-modulating modulus
  \[ \mu(\varphi)(n, \star) :=
  \begin{cases}
    (\star) & \text{if } n = 0 \\
    \epsilon & \text{otherwise.}
  \end{cases} \]
  Then $\usefirst\mu(\varphi)(n,\star) = (\star)$ and if the union is replaced by only using the last element we obtain the function that for $n=0$ reproduces $\mu$ and for $n>0$ returns $\epsilon$ if $\varphi \equiv \false$ and $(\false, \true)$ otherwise and can be checked to neither modulate $\usefirst M$ nor itself.
\end{example}

\subsection{Composition of monotone machines}\label{sec: monotone machines}
Monotone machines are easier to operate on as it is not necessary to keep track of the exact value of an effort that leads to a return value but an upper bound is sufficient.
Operations on monotone machines can usually be written down in a very straightforward manner.
As an example for this, let us describe the composition of two monotone machines in some detail.
Let $(M, \mu)$ and $(M', \mu')$ be monotone machines such that $F_M \colon \A^\Q \mto \A'^{\Q'}$ and $F_{M'}\colon \A'^{\Q'} \mto \A''^{\Q''}$.
Define the \demph{monotone machine composition} as follows:
First fix some default element $a'_d \in \A'$ and for each function $\varphi\colon \Q \to \A$ define a sequence of functions $\varphi'_n\colon \Q' \to \A'$ by
\[ \varphi'_n(q') :=
\begin{cases}
  a' & \text{if } M(\varphi)(n, q') = \some a' \\
  a'_d &\text{otherwise.}
\end{cases} \]
Use $\dom_n$ as shorthand for the set of elements $q' \in \Q'$ such that there exists an $a'$ with $M(\varphi)(n, q') = \some a'$.
Note that whenever $\varphi \in \dom(F_M)$, then $\varphi_n'$ and $F_M(\varphi)$ coincide on $\dom_n$ by these definitions.
Set
\[ (M'\circ_{\mu'} M)(\varphi)(n, q'') :=
\begin{cases}
  M'(\varphi'_n)(n, q'') & \text{if }\mu'(\varphi'_n)(n, q'') \subseteq \dom_n \\
  \none & \text{otherwise.}
\end{cases}
\]
Here, we put the index $\mu'$ at the composition as the outcome may be different for different valid moduli $\mu'$ of $M'$.
Define the composition of the moduli by
\[ (\mu\circ_M \mu')(\varphi)(n, q'') := \bigcup_{q' \in \mu'(\varphi'_n)(n, q'')} \mu(\varphi)(n, q'). \]
Just like the composition of machines depends on $\mu'$, the composition of moduli depends on $M$ via the definition $\varphi'_n$.
The above correctly implements composition:
\begin{theorem}
  If $(M, \mu)$ and $(M', \mu')$ are monotone machines, then so is $(M' \circ_{\mu'} M, \mu \circ_M \mu')$.
  Furthermore $F_{M' \circ_{\mu'} M} = F_{M'} \circ F_{M}$.
\end{theorem}
\begin{proof}  
  Let us first argue that the composition is monotone again.
  For this fix some inputs $\varphi$ and $q''$ and assume that $(M'\circ_{\mu'} M)(\varphi)(n,q'') = \some a''$.
  This can only be the case if $\mu'(\varphi'_n)(n,q'') \subseteq \dom_n$ and $M'(\varphi'_n)(n, q'') = \some a''$.
  To prove monotonicity we need to show that the same is true if $n$ is replaced by $n+1$.
  Since $M'$ is monotone it is sufficient to prove the list returned by the modulus to be included in $\dom_{n+1}$.
  Since $M$ is monotone, $\varphi'_n$ and $\varphi'_{n+1}$ coincide on $\dom_n$.
  As $\mu'$ is a modulus of $M'$, it holds that $M'(\varphi'_{n+1})(n,q'') = \some a''$.
  Since $\mu'$ terminates with $M'$, we get $\mu'(\varphi'_{n+1})(n+1, q'') = \mu'(\varphi'_{n+1})(n,q'')$.
  Finally, using that $\mu$ is self-modulating, we conclude
  \begin{align*}
    \mu'(\varphi'_{n+1})(n+1, q'') & = \mu'(\varphi'_{n+1})(n, q'') \\
    & = \mu'(\varphi'_n)(n, q'') \subseteq \dom_n \subseteq \dom_{n+1}.
  \end{align*}
  
  We omit the details of how to verify that the modulus is appropriate, and only outline how to prove the more important half of the equality, namely that the left-hand of the equation extends the right-hand side.
  For this assume that the right-hand side is defined in $\varphi$.
  This means that $\varphi \in \dom(F_{M})$ and $F_{M}(\varphi) \in \dom(F_{M'})$.
  Consider the sequence of functions $\varphi'_n$ as defined above and note that since $M$ is monotone and $\varphi\in \dom(F_{M})$, this sequence converges to $F_M(\varphi)$.
  Since $\mu'$ is self-modulating, it is in particular sequentially continuous and therefore the sequence $\mu'(\varphi'_n)$ converges to $\mu'(F_M(\varphi))$.
  This means that for any fixed $q''$ we can first pick $n$ big enough for $M'(F_M(\varphi))(n,q'')$ to take a value, then increase it further so that for all $k \geq n$ it holds that $\mu'(\varphi'_k)(n, q'') = \mu'(F_M(\varphi))(n, q'')$.
  As $\mu'$ terminates with $M'$, further increasing $n$ will no longer change the list it returns and we can use this to make sure that it is contained in $\dom_n$ as $\dom_n$ eventually contains every element of $\Q'$.
  As $q''$ was arbitrary, the left-hand side is defined and equal $F_{M'}(F_M(\varphi))$.
  
\end{proof}

Similar formulas can be found for other basic operations.
A composition of continuous machines can be obtained by first making these machines monotone and then composing them.
It is also possible compose continuous machines more directly but we failed to produce a simple description of this composition, the proofs of correctness are fairly involved and in experiments this composition did not perform well.
As a last remark, the formalization of the composition scheme, substitutes all uses of a default answer by default question.

\section{Conclusion}
Before we start a general discussion, let us comment on one very specific point.
While there are known constructions of self-modulating moduli from continuous moduli of continuity  \cite{fujiwara2019equivalence}, we are not aware of one that translates to our setting.
This is because all such constructions we are aware of directly work with the natural numbers and make use of their ordering.
We failed to recover such a result in our setting as our countability assumption seems too weak.
Being self-modulating turned out to be a very convenient property and most constructions of moduli that we came across result in self-modulating moduli independently of additional assumptions about the question and answer types.
Thus our decision to work with self-modulating moduli.

This paper is formulated from a point of view of computable analysis.
Computable analysis traditionally investigates known theorems from analysis and functional analysis concerning their computational content.
The mathematical background is developed over a classical meta-theory as are correctness proofs of algorithms.
An important part of computable analysis is that incomputable and discontinuous functions are not excluded and the classification of problems according to their degree of incomputability or discontinuity via Weihrauch reducibility is frequently studied \cite{brattka2017weihrauch}.
Our work and the \Incone{} library follow the traditions of computable analysis in the \Coq{} development and as mathematicians we found working over a strong meta-theory convenient.
A clear drawback is that providing computational content often means refining classical proofs and leads to some redundancy.
However, starting with a classical proof and effectivize step by step is often instructive.

Represented spaces relate to concepts popular in constructive analysis:
A representation defines a partial equivalence relation on its names by $\varphi \sim \psi \iff \delta(\varphi) = \delta(\psi)$.
Conversely, given a partial equivalence relation on Baire space one can consider the quotient space and consider the quotient mapping a representation.
Formulating everything using the equivalence relations, mentioning $X$ can be avoided completely.
This approach is for instance followed by developments like \corn{} \cite{cruz2004c}.
A function is called a morphism if it respects the equivalence relations and each such function induces a corresponding function on the equivalence classes that it realizes with respect to the quotient mapping as representation.
\Coq{} does not support quotient types and a direct description of the set of equivalence classes is additional information.
Definability of a function on abstract description need no longer correspond to computability.
Variations of this approach exist in \Coq{} and other proof assistants under the name ``refinements'' \cite{cohen2013refinements, lammich2013automatic}, but the objectives and with them which concepts are considered basic or useful differ significantly from our setting.

In our presentation we completely skipped the discussion of dialogue trees and jumped to associates directly.
In work about total functionals and mathematical work, dialogue trees play a central role.
Partial functions can be captured using a coinductive type of such trees.
We decided against this due to negative experiences with coinduction in \Coq{}, but we may try in the future.
It may also be worth looking into sequentiality concerns closer: While continuous machines characterize a sequential model of computation, they are seemingly non-sequential as the computations for different efforts may take distinct paths and need not be increasing in any way.
\bibliography{cc}{}

\newcommand{\etalchar}[1]{$^{#1}$}
\begin{thebibliography}{CFGW04}

\bibitem[ACM{\etalchar{+}}19]{affeldtmathematical}
Reynald Affeldt, Cyril Cohen, Assia Mahboubi, Damien Rouhling, and Pierre-Yves
  Strub.
\newblock The mathematical components analysis library, 2019.

\bibitem[ADK17]{altenkirch2017partiality}
Thorsten Altenkirch, Nils~Anders Danielsson, and Nicolai Kraus.
\newblock Partiality, revisited.
\newblock In {\em International Conference on Foundations of Software Science
  and Computation Structures}, pages 534--549. Springer, 2017.

\bibitem[Bau00]{Bauer:2000:RAC:933370}
Andrej Bauer.
\newblock {\em The Realizability Approach to Computable Analysis and Topology}.
\newblock PhD thesis, Carnegie Mellon University, Pittsburgh, PA, USA, 2000.
\newblock AAI3002721.

\bibitem[BG18]{brattka2018weihrauch}
Vasco Brattka and Guido Gherardi.
\newblock Weihrauch goes brouwerian.
\newblock {\em arXiv preprint arXiv:1809.00380}, 2018.

\bibitem[BG19]{brattka2019completion}
Vasco Brattka and Guido Gherardi.
\newblock Completion of choice.
\newblock {\em arXiv preprint arXiv:1910.13186}, 2019.

\bibitem[BGP17]{brattka2017weihrauch}
Vasco Brattka, Guido Gherardi, and Arno Pauly.
\newblock Weihrauch complexity in computable analysis.
\newblock {\em arXiv preprint arXiv:1707.03202}, 2017.

\bibitem[BKS16]{bove2016partiality}
Ana Bove, Alexander Krauss, and Matthieu Sozeau.
\newblock Partiality and recursion in interactive theorem provers--an overview.
\newblock {\em Mathematical Structures in Computer Science}, 26(1):38--88,
  2016.

\bibitem[Cap05]{Capretta05}
Venanzio Capretta.
\newblock General recursion via coinductive types.
\newblock {\em Logical Methods in Computer Science}, 1(2), 2005.

\bibitem[CDM13]{cohen2013refinements}
Cyril Cohen, Maxime D{\'e}n{\`e}s, and Anders M{\"o}rtberg.
\newblock Refinements for free!
\newblock In {\em International Conference on Certified Programs and Proofs},
  pages 147--162. Springer, 2013.

\bibitem[CFGW04]{cruz2004c}
Lu{\'\i}s Cruz-Filipe, Herman Geuvers, and Freek Wiedijk.
\newblock {C-CoRN}, the constructive {C}oq repository at {N}ijmegen.
\newblock In {\em International Conference on Mathematical Knowledge
  Management}, pages 88--103. Springer, 2004.

\bibitem[Con73]{constable1973type}
Robert~L Constable.
\newblock Type two computational complexity.
\newblock In {\em Proceedings of the fifth annual ACM symposium on Theory of
  computing}, pages 108--121. ACM, 1973.

\bibitem[Dzh19]{dzhafarov2017joins}
Damir~D. Dzhafarov.
\newblock Joins in the strong {W}eihrauch degrees.
\newblock {\em Math. Res. Lett.}, 26(3):749--767, 2019.

\bibitem[EX16]{ESCARDO2016770}
Mart^^c3^^adn Escard^^c3^^b3 and Chuangjie Xu.
\newblock A constructive manifestation of the {K}leene^^e2^^80^^93{K}reisel
  continuous functionals.
\newblock {\em Annals of Pure and Applied Logic}, 167(9):770 -- 793, 2016.
\newblock Fourth Workshop on Formal Topology (4WFTop).

\bibitem[FK19]{fujiwara2019equivalence}
Makoto Fujiwara and Tatsuji Kawai.
\newblock Equivalence of bar induction and bar recursion for continuous
  functions with continuous moduli.
\newblock {\em Annals of Pure and Applied Logic}, 170(8):867--890, 2019.

\bibitem[FS18]{ForsterSmolka:2018:Computability-JAR}
Yannick Forster and Gert Smolka.
\newblock {C}all-by-{V}alue {L}ambda {C}alculus as a {M}odel of {C}omputation
  in {C}oq.
\newblock {\em Journal of Automated Reasoning}, 2018.

\bibitem[Grz57]{MR0089809}
Andrzej Grzegorczyk.
\newblock On the definitions of computable real continuous functions.
\newblock {\em Fund. Math.}, 44:61--71, 1957.

\bibitem[HEX15]{hotzel2015inconsistency}
Mart{\'\i}n H{\"o}tzel~Escard{\'o} and Chuangjie Xu.
\newblock The inconsistency of a brouwerian continuity principle with the
  curry--howard interpretation.
\newblock In {\em 13th International Conference on Typed Lambda Calculi and
  Applications (TLCA 2015)}. Schloss Dagstuhl-Leibniz-Zentrum fuer Informatik,
  2015.

\bibitem[Kaw11]{kawamuraphd}
Akitoshi Kawamura.
\newblock {\em Computational complexity in analysis and geometry}.
\newblock University of Toronto, 2011.

\bibitem[KC12]{kawamura_complexity_2012}
Akitoshi Kawamura and Stephen Cook.
\newblock Complexity theory for operators in analysis.
\newblock {\em ACM Transactions in Computation Theory}, 4(2):Article 5, 2012.

\bibitem[Kle53]{kleene1953recursive}
Stephen~Cole Kleene.
\newblock Recursive functions and intuitionistic mathematics.
\newblock 1953.

\bibitem[Kle59]{kleene1959constructivity}
Stephen~Cole Kleene.
\newblock Countable functionals.
\newblock {\em Constructivity in Mathematics: proceedings of the colloquium
  held at Amsterdam}, 1959.

\bibitem[Ko91]{MR1137517}
Ker-I Ko.
\newblock {\em Complexity theory of real functions}.
\newblock Progress in Theoretical Computer Science. Birkh{\"a}user Boston Inc.,
  Boston, MA, 1991.

\bibitem[KP14]{kawamura2014function}
Akitoshi Kawamura and Arno Pauly.
\newblock Function spaces for second-order polynomial time.
\newblock In {\em Conference on Computability in Europe}, pages 245--254.
  Springer, 2014.

\bibitem[Kre59]{MR0106838}
Georg Kreisel.
\newblock Interpretation of analysis by means of constructive functionals of
  finite types.
\newblock In {\em Constructivity in mathematics: {P}roceedings of the
  colloquium held at {A}msterdam, 1957 (edited by {A}. {H}eyting)}, Studies in
  Logic and the Foundations of Mathematics, pages 101--128. North-Holland
  Publishing Co., Amsterdam, 1959.

\bibitem[KW85]{kreitz1985theory}
Christoph Kreitz and Klaus Weihrauch.
\newblock Theory of representations.
\newblock {\em Theoretical computer science}, 38:35--53, 1985.

\bibitem[Lac58]{MR0105357}
Daniel Lacombe.
\newblock Sur les possibilit\'es d'extension de la notion de fonction
  r\'ecursive aux fonctions d'une ou plusieurs variables r\'eelles.
\newblock In {\em Le raisonnement en math\'ematiques et en sciences
  exp\'erimentales}, Colloques Internationaux du Centre National de la
  Recherche Scientifique, LXX, pages 67--75. Editions du Centre National de la
  Recherche Scientifique, Paris, 1958.

\bibitem[Lam06]{lambov2006basic}
Branimir Lambov.
\newblock The basic feasible functionals in computable analysis.
\newblock {\em Journal of Complexity}, 22(6):909--917, 2006.

\bibitem[Lam13]{lammich2013automatic}
Peter Lammich.
\newblock Automatic data refinement.
\newblock In {\em International Conference on Interactive Theorem Proving},
  pages 84--99. Springer, 2013.

\bibitem[LN15]{longley2015higher}
John Longley and Dag Normann.
\newblock {\em Higher-order computability}, volume 100.
\newblock Springer, 2015.

\bibitem[Lon02]{longley2002sequentially}
John Longley.
\newblock The sequentially realizable functionals.
\newblock {\em Annals of Pure and Applied Logic}, 117(1-3):1--93, 2002.

\bibitem[MT17]{mahboubi2017mathematical}
Assia Mahboubi and Enrico Tassi.
\newblock Mathematical components, 2017.

\bibitem[Nor00]{normann2000computability}
Dag Normann.
\newblock Computability over the partial continuous functionals.
\newblock {\em The Journal of Symbolic Logic}, 65(3):1133--1142, 2000.

\bibitem[NS19]{NEUMANN2019}
Eike Neumann and Florian Steinberg.
\newblock Parametrised second-order complexity theory with applications to the
  study of interval computation.
\newblock {\em Theoretical Computer Science}, 2019.

\bibitem[Pau16]{pauly2016topological}
Arno Pauly.
\newblock On the topological aspects of the theory of represented spaces.
\newblock {\em Computability}, 5(2):159--180, 2016.

\bibitem[PZ13]{Pauly2013RelativeCA}
Arno Pauly and Martin Ziegler.
\newblock Relative computability and uniform continuity of relations.
\newblock {\em J. Logic \& Analysis}, 5, 2013.

\bibitem[SAB{\etalchar{+}}19]{metacoq2019hal}
Matthieu Sozeau, Abhishek Anand, Simon Boulier, Cyril Cohen, Yannick Forster,
  Fabian Kunze, Gregory Malecha, Nicolas Tabareau, and Th{\'e}o Winterhalter.
\newblock {The MetaCoq Project}.
\newblock working paper or preprint, June 2019.

\bibitem[Sch02]{schroeder_phd}
Matthias Schr{\"o}eder.
\newblock {\em Admissible Representations for Continuous Computations}.
\newblock PhD thesis, FernUniversit{\"a}t Hagen, 2002.

\bibitem[Sch04]{schroedert2}
Matthias Schr^^c3^^b6der.
\newblock Spaces allowing type-2 complexity theory revisited.
\newblock {\em Math. Log. Q.}, 50:443--459, 09 2004.

\bibitem[Ste17]{steinberg2017computational}
Florian Steinberg.
\newblock {\em Computational Complexity Theory for Advanced Function Spaces in
  Analysis}.
\newblock PhD thesis, Technische Universit{\"a}t, 2017.

\bibitem[STT19]{steinberg2019quantitative}
Florian Steinberg, Laurent Thery, and Holger Thies.
\newblock Quantitative continuity and computable analysis in coq.
\newblock {\em arXiv preprint arXiv:1904.13203}, 2019.

\bibitem[Tur36]{turing_computable_1936}
Alan~Mathison Turing.
\newblock On computable numbers, with an application to the
  {Entscheidungsproblem}.
\newblock {\em Proceedings of the London Mathematical Society}, 2(1):230--265,
  1936.

\bibitem[TvD88]{MR966421}
Anne~Sjerp Troelstra and Dirk van Dalen.
\newblock {\em Constructivism in mathematics. {V}ol. {II}}, volume 123 of {\em
  Studies in Logic and the Foundations of Mathematics}.
\newblock North-Holland Publishing Co., Amsterdam, 1988.

\bibitem[VO{\etalchar{+}}11]{van2011partial}
Jaap Van~Oosten et~al.
\newblock Partial combinatory algebras of functions.
\newblock {\em Notre Dame Journal of formal logic}, 52(4):431--448, 2011.

\bibitem[Wei00]{weihrauch_computable_2000}
Klaus Weihrauch.
\newblock {\em Computable {Analysis}}.
\newblock Springer, Berlin/Heidelberg, 2000.

\bibitem[XE13]{xu2013constructive}
Chuangjie Xu and Mart{\'\i}n Escard{\'o}.
\newblock A constructive model of uniform continuity.
\newblock In {\em International Conference on Typed Lambda Calculi and
  Applications}, pages 236--249. Springer, 2013.

\end{thebibliography}
\end{document}